\documentclass[10pt,,draftcls,a4paper,oneside,onecolumn,peerreview]{IEEEtran}
\usepackage[USenglish]{babel} %francais, polish, spanish, ...
\usepackage[T1]{fontenc}
\usepackage[ansinew]{inputenc}
\usepackage{lmodern}

\usepackage[noadjust]{cite}

\usepackage{fancyhdr}
\usepackage{lastpage}
\usepackage{textcomp}
\usepackage{xcolor}

\usepackage{amsmath}
\usepackage{amsthm}
\usepackage{amsfonts}
\usepackage{dsfont} % $\mathds{1}$
\usepackage{stmaryrd} % llbracket

\usepackage{graphicx}
\usepackage{multirow}
\usepackage[bookmarks = true, colorlinks = true]{hyperref}

\newtheorem{theorem}{Theorem}
\newtheorem{definition}{Definition}

\newtheorem{corollary}{Corollary}
\newtheorem{lemma}{Lemma}
\newtheorem{proposition}{Proposition}

\theoremstyle{definition}
\newtheorem{remark}{Remark}
\newtheorem{example}{Example}

\DeclareMathOperator*{\argsup}{arg\,sup\,}
\DeclareMathOperator*{\arginf}{arg\,inf\,}
\begin{document}

\title{Semiparametric two-component mixture models when one component is defined through linear constraints}
\author{Diaa~Al~Mohamad \; \; \; Assia~BOUMAHDAF,%
\thanks{Diaa Al Mohamad is a PhD. student at Laboratoire de Statistique Th\'eorique et Appliqu\'ee at the Univeristy of Paris 6 (UPMC) 4 place Jussieu 75005 Paris - France. Diaa Al Mohamad is the corresponding author of the article. email: diaa.almohamad@gmail.com}%
\thanks{Assia BOUMAHDAF is a PhD. student at Laboratoire de Statistique Th\'eorique et Appliqu\'ee at the Univeristy of Paris 6 (UPMC) 4 place Jussieu 75005 Paris - France. Email: assia.boumahdaf@gmail.com}}

%\IEEEspecialpapernotice{A part of this work was presented in the conference paper \cite{AlMohamad2015}}

\maketitle

%\pagestyle{fancy} %Now display headings: headings / fancy / ...
%\rhead{}
%\lhead{\footnotesize{Proximal Point Algorithm for MDE with Application to Mixture Models}}
%\chead{}
%\cfoot{\thepage\ / \pageref{LastPage}}
%\lfoot{\footnotesize{Diaa AL MOHAMAD}}
%\rfoot{\footnotesize{Last update \today}}
\begin{abstract}
We propose a structure of a semiparametric two-component mixture model when one component is parametric and the other is defined through linear constraints on its distribution function. Estimation of a two-component mixture model with an unknown component is very difficult when no particular assumption is made on the structure of the unknown component. A symmetry assumption was used in the literature to simplify the estimation. Such method has the advantage of producing consistent and asymptotically normal estimators, and identifiability of the semiparametric mixture model becomes tractable. Still, existing methods which estimate a semiparametric mixture model have their limits when the parametric component has unknown parameters or the proportion of the parametric part is either very high or very low. We propose in this paper a method to incorporate a prior linear information about the distribution of the unknown component in order to better estimate the model when existing estimation methods fail. The new method is based on $\varphi-$divergences and has an original form since the minimization is carried over both arguments of the divergence. The resulting estimators are proved to be consistent and asymptotically normal under standard assumptions. We show that using the Pearson's $\chi^2$ divergence our algorithm has a linear complexity when the constraints are moment-type. Simulations on univariate and multivariate mixtures demonstrate the viability and the interest of our novel approach.
\end{abstract}

\IEEEpeerreviewmaketitle
%--------------------------------------------------------------------------------------
\section*{Introduction}
\IEEEPARstart{A} two-component mixture model with an unknown component is defined by:
\begin{equation}
f(x) = \lambda f_1(x|\theta) + (1-\lambda) f_0(x), \qquad \text{for } x\in\mathbb{R}^r
\label{eqn:GeneralSemiParaMix}
\end{equation}
with $\lambda\in(0,1)$ and $\theta\in\mathbb{R}^d$ to be estimated as the density $f_0$ is unknown. This model appears in the study of gene expression data coming from microarray analysis. An application to two bovine gestation mode comparison is performed in \cite{Bordes06b}. The authors suppose that $\theta$ is known, $f_0$ is symmetric around an unknown $\mu$ and that $r=1$. In \cite{Xiang}, the authors studied a more general setup by considering $\theta$ unknown and applied model (\ref{eqn:GeneralSemiParaMix}) on the Iris data by considering only the first principle component for each observed vector. Another application of model (\ref{eqn:GeneralSemiParaMix}) in genetics can be found in \cite{JunMaDiscret}. The semiparametric model was employed in \cite{Robin} (supposing that $\theta$ is known) for multiple testing procedures in order to estimate the posterior population probabilities and the local false rate discovery. In \cite{Song}, the authors studied a similar setup where $\theta$ is unknown without further assumptions on $f_0$. They applied the semiparametric model in sequential clustering algorithms as a second step after having calculated the centers of the clusters. The model of the current cluster $f_1$ is Gaussian with known location and unknown scale, whereas the distribution of the remaining clusters is represented by $f_0$ and is supposed to be unknown. Finally, Model (\ref{eqn:GeneralSemiParaMix}) can also be regarded as a contamination model, see \cite{Titterington} or \cite{FinMixModMclachlan} for further applications of mixture models. \\

Several estimation methods were proposed in the aforementioned papers. A method was proposed in \cite{Bordes06b} and was later improved in \cite{Bordes10} where the authors assume that the unknown density is symmetric around an unknown value $\mu$ and that the parametric component is fully known, i.e. $\theta$ is known.
 %The idea behind their procedure is to calculate $f_0$ as a function of other terms in equation (\ref{eqn:GeneralSemiParaMix}), then use the symmetry of $f_0$ to write $\mathbb{F}_0(x)=1-\mathbb{F}_0(-x)$, where $\mathbb{F}_0$ is the cumulative distribution function of $f_0$. We then use a suitable distance to compare between $\mathbb{F}_0(x)$ and $1-\mathbb{F}_0(-x)$.
The method as it is cannot be used in multivariate contexts. Besides, in univariate ones it is not possible to use such method when any of the two components of the mixture has a nonnegative support. In comparison to other proposed methods, this method has the advantage of having a solid theoretical basis. The resulting estimators are proved in \cite{Bordes10} to be consistent and asymptotically Gaussian. Besides, the authors prove (see also \cite{Bordes06b}) that the model becomes identifiable under further assumptions on the parametric component. Two other estimation methods were proposed in \cite{Song}; the $\pi-$maximizing method and an EM-type algorithm. The $\pi-$maximizing method is based on the identifiability of model (\ref{eqn:GeneralSemiParaMix}) when $f_1$ is a scale Gaussian model. Asymptotic properties of this method were not studied and theoretical justification is only done in the Gaussian case and it is not clear how to generalize it. The method is still adaptable to multivariate situations. Their EM-type algorithm shares similarities with other exiting approaches in the literature such as \cite{Robin}, \cite{JunMaDiscret} and \cite{BordesStochEM}. These algorithms estimate at each iteration the proportion of the parametric component as an average of weights attributed to the observations. The difference between these methods is in their way of calculating the vector of weights. EM-type methods are not based on the minimization of a criterion function as in \cite{Bordes10} or \cite{Song}. Besides, their asymptotic properties are generally very difficult to establish. Finally, a Hellinger-based two-step directional optimization procedure was proposed in \cite{Xiang}; a first step minimizes the divergence over $f_0$ and the second step minimizes over the parameters $(\lambda,\theta)$. Their method seems to give good results, but the algorithm is very complicated and no explanation on how to do the calculation is given. Properties of the iterative procedure are not studied either.\\ 

All above methods were illustrated to work in specific situations when the parametric component is fully known or is a Gaussian distribution. A comparison between the three methods proposed in \cite{Bordes10} and \cite{Song} is illustrated in \cite{Xiang} and are comparted to the Hellinger-based iterative method of the same authors on a two-component Gaussian data provided that the parametric component is fully known. As we add $\theta$ to the set of unknown parameters, i.e. the parametric component is not fully known, things become different. The symmetry method of \cite{Bordes10} does not perform well unless the proportion of the unknown component $1-\lambda$ is high enough. This is not surprising since this method is based on the properties of the unknown component; hence it should be well estimated. On the contrary, other methods (EM-type methods and the $\pi-$maximizing one) perform well when the proportion of the parametric component $\lambda$ is high enough.\\
It is important and of further interest that the estimation method takes into account possible unknown parameters in the parametric component $f_1$. We believe that the failure of the existing methods to treat model (\ref{eqn:GeneralSemiParaMix}) comes from the degree of difficulty of the semiparametric model itself. The use of a symmetric assumption made the estimation better in some contexts, but such assumption is still restrictive and cannot be applied on positive-supported distributions. We need to incorporate other prior information about $f_0$ in a way that we stay in between a fully parametric settings and a fully semiparametric one. We thus propose a method which permits to add relatively general prior information. Such information needs to apply linearly on the distribution function of the unknown component such as moment-type information. For example, we may have an information relating the first and the second moments of $f_0$ such as $\int{xf_0(x)dx}=\alpha$ and $\int{x^2f_0(x)dx} = m(\alpha)$ for some unknown $\alpha$, see \cite{BroniaKeziou12} and the references therein. Such information adds some structure to the model without precising the value of the moments. More examples will be discussed later on.\\
Unfortunately, the incorporation of linear constraints on the distribution function cannot be done directly in existing methods because the optimization will be carried over a (possibly) infinite dimensional space, and we need a new tool. Convex analysis offers a way using the Fenchel-Legendre duality to transform an optimization problem over an infinite dimensional space into the space of Lagrangian parameters (finite dimensional one). $\varphi-$divergences offer a way by their convexity properties to use this duality result. A complete study of this problem in the nonmixture case is done in \cite{BroniaKeziou12}; see also \cite{AlexisGSI13} and \cite{KeziouThesis} Chap. 3. We will exploit the results in these papers to build upon a new estimation procedure which takes into account linear information over the unknown component's distribution.\\

The paper is organized as follows. Section \ref{sec:GeneralModelDef} presents the general context of semiparametric models under linear constraints. Section \ref{sec:phiDivergenceDef} presents $\varphi-$divergences and some of their general properties. We show how we can estimate a semiparametric model using $\varphi-$divergences. Section \ref{sec:SemiparaMixturesConstr} introduces the semiparametric mixture model (\ref{eqn:GeneralSemiParaMix}) when $f_0$ is defined through linear constraints. We introduce an algorithm which permits to estimate efficiently using $\varphi-$divergences the semiparametric mixture model. Identifiability and existence of a unique solution of the estimation methodology are also discussed. In Section \ref{sec:AsymptotResults}, we prove that our estimator is consistent and asymptotically normal under standard assumptions. Finally, Section \ref{sec:SemiParaSimulations} is devoted to the illustration of the method on several mixture models in univariate and multivarite contexts and a comparison with existing methods which permits to show how the prior information can improve the estimation.

%%%%%%%%%%%%%%%%%%%%%%%%%%%%%%%%%%%%%%%%%%%%%%%%%%%%%%%%%%%%%%%%%%
%%%%%%%%%%%%%%%%%%%%%%%%%%%%%%%%%%%%%%%%%%%%%%%%%%%%%%%%%%%%%%%%%%

\section{Models defined through linear constraints on the distribution function}\label{sec:GeneralModelDef}
We want to integrate linear information in the semiparametric mixture model (\ref{eqn:GeneralSemiParaMix}). The linear information is a set of linear constraints imposed on the unknown component. The objective is still to retrieve the true vector of parameters defining the model on the basis of a given i.i.d. sample $X_1,\cdots,X_n$ drawn from the mixture distribution $P_T$. \\
We prefer to proceed step by step. We start with models which can be defined through a linear information. These models are not necessarily mixtures of distributions in this section. Besides, the constraints or the linear information defining the model will apply over the whole model, i.e. if the model is a mixture then the constraints apply over the whole mixture and not only over one component. We give in this section a brief idea of what the literature offers us to study and estimate such model. In the next section we will proceed to aggregate the two ideas, i.e. mixture models and semiparametric models defined through linear constraints, in order to introduce our semiparametric mixture model where a component is parametric (but not fully known) and a component is defined through linear constraints, hence semiparametric with unknown parameters.
\subsection{Definition and examples}
Denote by $M^{+}$ the set of all probability measures (p.m.) defined on the same measurable space as $P_T$, i.e. $(\mathbb{R}^{r},\mathcal{B}(\mathbb{R}^{r}))$.
\begin{definition}
Let $X_1,\cdots,X_n$ be random variables drawn independently from the probability distribution $P_T$. A semiparametric linear model is a collection of probability measures $\mathcal{M}_{\alpha}(P_T)$, for $\alpha\in\mathcal{A}\subset\mathbb{R}^s$, absolutely continuous with respect to $P_T$ which verifies a set of linear constraints, i.e.
\begin{equation}
\mathcal{M}_{\alpha}(P_T) = \left\{ Q\in M^{+} \; \; \textrm{such that} \;\; Q \ll P_T, \;\; \int g(x)dQ(x) = m(\alpha)\right\},
\label{eqn:SimpleConstraints}
\end{equation}
where $g : \mathbb{R}^{r} \rightarrow \mathbb{R}^{\ell}$ and $m:\mathcal{A}\rightarrow \mathbb{R}^{\ell}$ are specified vector-valued functions.
\end{definition}
\noindent This semiparametric linear model was studied by many authors, see \cite{BroniaKeziou12}, \cite{AlexisGSI13} (with $dP$ replaced by $dP^{-1}$ the quantile measure), \cite{Owen} (in the empirical likelihood context) and \cite{ChenQin} (for finite population problems). It is possible in the above definition to make $g$ depend on the parameter vector $\alpha$, but we stay for the sake of simplicity with the assumption that $g$ does not depend on $\alpha$. The theoretical approach we present in this paper remains valid if $g$ depends on $\alpha$ with slight modifications on the assumptions and more technicalities at the level of the proofs.
\begin{example}
A simple and standard example is a model defined through moment constraints. Let $P_T$ be the Weibull distribution with scale $a^*$ and shape $b^*$. We define $\mathcal{M}_{\alpha}$ with $\alpha = (a,b)\in(0,\infty)^2$ to be the set of all probability measures whose first three moments are given by:
\[\int{x^{i}dQ(x)} = a^{i}\Gamma(1+i/b), \qquad i=1,2,3.\]
The set $\mathcal{M}_{\alpha^*}$ is a "neighborhood" of probability measures of the Weibull distribution $P_T$. It contains all probability measures absolutely continuous with respect to the Weibull mixture $P_T$ and which share the first three moments with it. The union of the sets $\mathcal{M}_{\alpha}$ contains all probability measures whose first three moments share the same analytic form as a Weibull distribution.
\end{example}
If the true distribution $P_T$ verifies the set of $\ell$ constraints (\ref{eqn:SimpleConstraints}) for some $\alpha^*$, then the set 
\begin{equation}
\mathcal{M}_{\alpha^*}(P_T) = \left\lbrace Q\in M^{+} \; \; \textrm{such that} \;\; Q \ll P_T, \;\; \int g(x)dQ(x) = m(\alpha^*) \right\rbrace
\label{eqn:SetMProbaConstr}
\end{equation}
constitutes a "neighborhood" of probability measures of $P_T$. Generally, one would rather consider the larger "neighborhood" defined by
\begin{equation*}
\mathcal{M} = \bigcup_{\alpha \in \mathcal{A}} \mathcal{M}_{\alpha},
\end{equation*}
because the value of $\alpha^*$ is unknown and needs to be estimated. The estimation procedure aims at finding $\alpha^*$ the ("best") vector for which $P_T\in M_{\alpha^*}$. This is generally done by either solving the set of equations (\ref{eqn:SimpleConstraints}) defining the constraints for $Q$ replaced by (an estimate of) $P_T$ or by minimizing a suitable distance-like function between the set $\mathcal{M}$ and (an estimate of) $P_T$. In other words, we search for the "projection" of $P_T$ on $\mathcal{M}$. Solving the set of equations (\ref{eqn:SimpleConstraints}) is in general a difficult task since it is a set of nonlinear equations. In the literature, similar problems were solved using the Fenchel-Legendre duality, see \cite{BroniaKeziou12} and \cite{AlexisGSI13}.\\

In the next section, we present the approach proposed in \cite{BroniaKeziou12} which will be the basis for estimating our semiparametric mixture model. We recall the notion of a $\varphi-$divergence and present the duality result which permits to estimate a semiparametric model efficiently.

%%%%%%%%%%%%%%%%%%%%%%%%%%%%%%%%%%%%%%%%%%%%%%%%%%%%%%%%%%%%%%
%
% ==========================================================================
% ==========================================================================
%
%%%%%%%%%%%%%%%%%%%%%%%%%%%%%%%%%%%%%%%%%%%%%%%%%%%%%%%%%%%%%%
\section{Estimation of a semiparametric model using \texorpdfstring{$\varphi$}{phi}-divergences and the duality technique}\label{sec:phiDivergenceDef}
\subsection{\texorpdfstring{$\varphi$}{phi}-divergences: Definitions and general properties}\label{subsec:DefPropPhiDiverg}
$\varphi$-divergences are measures of dissimilarities. They were introduced independently in \cite{Csiszar1963} (as "$f$-divergences") and \cite{AliSilvey}. Let $P$ be a probability measure and $Q$ be a finite signed measure defined on $(\mathbb{R}^{r}, \mathcal{B}(\mathbb{R}^{r}))$ such that $Q$ is absolutely continuous (a.c.) with respect to (w.r.t.) $P$. Let $\varphi : \mathbb{R} \mapsto [0, +\infty]$ be a proper convex function with $\varphi(1) = 0$ and such that its domain $\textrm{dom}\varphi = \left\lbrace   x \in \mathbb{R} \;\; \textrm{such that} \;\; \varphi(x) < \infty \right\rbrace := (a_{\varphi},b_{\varphi})$ with $a_{\varphi} < 1 < b_{\varphi}$. The $\varphi$-divergence between $Q$ and $P$ is defined by

\begin{equation*}
D_{\varphi}(Q,P) = \int_{\mathbb{R}^r}{ \varphi\left( \dfrac{dQ}{dP}(x) \right)dP(x)},
\end{equation*}

\noindent
where $\dfrac{dQ}{dP}$ is the Radon-Nikodym derivative. When $Q$ is not a.c.w.t. $P$, we set $D_{\varphi}(Q,P) = + \infty$. When, $P = Q$ then $D_{\varphi}(Q,P) = 0$. Furthermore, if the function $x \mapsto \varphi(x)$ is strictly convex on a neighborhood of $x=1$, then

\begin{equation}
\label{fondamental property of divergence}
D_{\varphi}(Q,P) = 0 \;\; \textrm{ if and only if} \;\; P = Q.
\end{equation}

\noindent Several standard statistical divergences can be expressed as $\varphi-$divergences; the Hellinger, the Pearson's and the Neymann's $\chi^2$, and the (modified) Kullback-Leibler. They all belong to the class of Cressie-Read, also known as "power divergences", (see \cite{CressieRead1984}). It is defined by
\begin{equation}
\varphi_{\gamma}(x) := \dfrac{x^{\gamma}-\gamma x + \gamma -1}{\gamma(\gamma -1)},
\label{eqn:CressieReadPhi}
\end{equation}
for $\gamma=\frac{1}{2},2,-2,0,1$ respectively. For $\gamma=0$, we denote $\varphi_0(t)=-\log(t)+t-1$ for the likelihood divergence and when $\gamma=1$, we denote $\varphi_1(t)=t\log(t)-t+1$ for the Kullback-Leibler. More details and properties can be found in \cite{LieseVajda} or \cite{Pardo}. \\
Estimators based on such a tool were developed in the parametric (see \cite{Beran},\cite{LindsayRAF},\cite{ParkBasu},\cite{BroniaKeziou09}) and the semiparametric setups (see \cite{BroniaKeziou12} and \cite{AlexisGSI13}). In all these methods, the $\varphi-$divergence is calculated between a model $Q$ and a true distribution $P_T$. We search our estimators by making the model approach\footnote{More accurately, we project the true distribution on the model.} the true distribution $P_T$. In this paper, we provide an original method where the minimization is done over both arguments of the divergence in a way that the two arguments approach one another for the sake of finding a suitable estimator\footnote{There still exists some work in computer vision using $\varphi-$divergences where the minimization is done over both arguments of the divergence, see \cite{Mireille}.  The work concerns a parametric setup in discrete models.}. In completely nonparametric setup, we may mention the work in \cite{KarunamuniWu} on two component mixture models when both components are unknown. The authors use the Hellinger divergence, and assume that we have in hand a sample from each component and a sample drawn from the whole mixture. For regression of nonparametric mixture models using the Hellinger divergence, see \cite{TangRegression}.\\
The following definitions concern the notion of $\varphi$-projection of finite signed measures over a set of finite signed measures.
\begin{definition}
\label{def:phiDistance}
Let $\mathcal{M}$ be some subset of $M$, the space of finite signed measures. The $\varphi$-divergence between the set $\mathcal{M}$ and some probability measure $P$, noted as $D_{\varphi}(\mathcal{M},P)$, is given by
\begin{equation}
  D_{\varphi}\left(\mathcal{M},P\right)  := \inf_{Q \in \mathcal{M}}D_{\varphi}(Q,P).
	\label{eqn:PhiDistanceEleSet}
\end{equation}
Furthermore, we define the $\varphi-$divergence between a subset of $M$, say $\mathcal{M}$, and a subset of $M^+$ (the set of all probability measures), say $\mathcal{N}$ by
\begin{equation*}
  D_{\varphi}\left(\mathcal{M},\mathcal{N}\right)  := \inf_{Q \in \mathcal{M}}\inf_{P\in\mathcal{N}}D_{\varphi}(Q,P).
\end{equation*}
\end{definition}
\begin{definition}
Assume that $D_{\varphi}(\mathcal{M},P)$ is finite. A measure $Q^* \in \mathcal{M}$ such that
\begin{equation*}
D_{\varphi}(Q^*,P) \leq D_{\varphi}(Q,P), \;\; \textrm{for all}\;\; Q \in \mathcal{M}
\end{equation*}
is called a $\varphi$-projection of $P$ onto $\mathcal{M}$. This projection may not exist, or may not be defined uniquely.
\end{definition}

The essential tool we need from a $\varphi-$divergence is its characterization of the projection of some probability measure $P$ onto a set $\mathcal{M}$ of finite signed measures. Such characterization will permit to transform the search of a projection in an infinite dimensional space to the search of a vector $\xi$ in $\mathbb{R}^{\ell}$.
%%%%%%%%%%%%%%%%%%%%%%%%%%%%%%%%%%%%%%%%%%%%%%%%%%%%%%%%%%%%%%%%%%%%%%%%%%%%
\subsection{Estimation of the semiparametric model using \texorpdfstring{$\varphi-$}{phi}divergences and the duality technique}\label{subsec:DualTechRes}
Estimation of the semiparametric model using $\varphi-$divergences is summarized by the following optimization problem:
\begin{equation}
\alpha^* = \arginf_{\alpha\in\mathcal{A}}\inf_{Q\in\mathcal{M}_{\alpha}} D_{\varphi}(Q,P_T).
\label{eqn:SemiParaEstimGeneral}
\end{equation}
We are, then, searching for the projection of $P_T$ on the set $\mathcal{M}=\cup_{\alpha}\mathcal{M}_{\alpha}$. We are more formally interested in the vector $\alpha^*$ for which the projection of $P_T$ on $\mathcal{M}$ belongs to the set $\mathcal{M}_{\alpha^*}$. Notice that the sets $\mathcal{M}_{\alpha}$ need to be disjoint so that the projection cannot belong to several sets in the same time. \\
The magical property of $\varphi-$divergences stems from their characterization of the projection of a probability measure $P$ onto a set $\mathcal{M}$ of finite signed measures, see \cite{BroniaKeziou2006} Theorem 3.4. Such characterization permits to transform the search of a projection in an infinite dimensional space into the search of a vector $\xi$ in $\mathbb{R}^{\ell}$ through the duality of Fenchel-Legendre and thus simplify the optimization problem (\ref{eqn:SemiParaEstimGeneral}). Note that Theorem 3.4 from \cite{BroniaKeziou2006} provides a formal characterization of the projection, but we will only use it implicitly.\\ 
Let $\varphi$ be a strictly convex function which verifies the same properties mentioned in the definition of a $\varphi-$divergence, see paragraph \ref{subsec:DefPropPhiDiverg}. The Fenchel-Legendre transform
of $\varphi$, say $\psi$ is defined by
\begin{equation*}
\psi(t) = \sup_{x \in \mathbb{R}}\left\lbrace tx - \varphi(x) \right\rbrace, \qquad  \forall t	\in\mathbb{R}.
\end{equation*}
\noindent We are concerned with the convex optimization problem
\begin{equation}
\label{primal problem}
(\mathcal{P}) \qquad \inf_{Q \in \mathcal{M}_{\alpha}} D_{\varphi}(Q,P_T).
\end{equation}
\noindent We associate to $(\mathcal{P})$ the following dual problem
\begin{equation}
\label{dual problem}
(\mathcal{P}^*) \qquad \sup_{\xi\in\mathbb{R}^{\ell}} \xi^t m(\alpha)  - \int{\psi\left(\xi^t g(x)\right) dP_T(x)}.
\end{equation}
\noindent We require that $\varphi$ is differentiable. Assume furthermore that $\int |g_i(x)|dP_T(x) < \infty$ for all $i=1,\ldots, \ell$ and there exists some measure $Q_T$ a.c.w.r.t. $P_T$ such that $D_{\varphi}(Q_T,P_T) < \infty$. According to Proposition 26 from \cite{AlexisGSI13} (see also Proposition 4.2 in \cite{BroniaKeziou12} or Theorem 5.1 in \cite{BroniaKeziouSigned}) we have a strong duality attainment, i.e. $(\mathcal{P}) = (\mathcal{P}^*)$. In other words,
\begin{equation}
\inf_{Q\in\mathcal{M}_{\alpha}} D_{\varphi}\left(Q,P_T\right)=\sup_{\xi\in\mathbb{R}^{\ell}} \xi^t m(\alpha)  - \int{\psi\left(\xi^t g(x)\right) dP_T(x)}.
\label{strong duality}
\end{equation}
The estimation procedure of the semiparametric model (\ref{eqn:SemiParaEstimGeneral}) is now simplified into the following finite-dimensional optimization problem
\begin{eqnarray*}
\alpha^* & = & \arginf_{\alpha\in\mathcal{A}} D_{\varphi}(\mathcal{M}_{\alpha},P_T) \\
  & = & \arginf_{\alpha\in\mathcal{A}} \sup_{\xi\in\mathbb{R}^{\ell}} \xi^t m(\alpha)  - \int{\psi\left(\xi^t g(x)\right) dP_T(x)}.
\end{eqnarray*}
This is indeed a feasible procedure since we only need to optimize a real function over $\mathbb{R}^{\ell}$. Examples of such procedures can be found in \cite{BroniaKeziou12}, \cite{AlexisGSI13}, \cite{NeweySmith} and the references therein. Robustness of this procedure was studied theoretically by \cite{TomaRobustLinConstr}.  \\

Now that all notions and analytical tools are presented, we proceed to the main objective of this chapter; semiparametric mixtures models. The following section defines such models and presents a method to estimate them using $\varphi-$divergences. We study after that the asymptotic properties of the vector of estimates.

%%%%%%%%%%%%%%%%%%%%%%%%%%%%%%%%%%%%%%%%%%%%%%%%%%%%%%%%%%%%%%%%%%%%%%%%%%%%%%
%%%%%%%%%%%%%%%%%%%%%%%%%%%%%%%%%%%%%%%%%%%%%%%%%%%%%%%%%%%%%%%%%%%%%%%%%%%%%%%%%%%%%%%%%%%
%%%%%%%%%%%%%%%%%%%%%%%%%%%%%%%%%%%%%%%%%%%%%%%%%%%%%%%%%%%%%%%%%%%%%%%%%%%%%%%%%%%%%%%%
%
\section{Semiparametric two-component mixture models when one component is defined through linear constraints}\label{sec:SemiparaMixturesConstr}
\subsection{Definition and estimating methodology}
\begin{definition} 
\label{def:SemiParaModel}
Let $X$ be a random variable taking values in $\mathbb{R}^{r}$ distributed from a probability measure $P$. We say that $P(.|\phi)$ with $\phi=(\lambda,\theta,\alpha)$ is a two-component semiparametric mixture model subject to linear constraints if it can be written as follows:
\begin{eqnarray}
P(.| \phi) & = &  \lambda P_1(.|\theta) + (1-\lambda) P_0 \quad \text{s.t. } \nonumber\\
P_0\in\mathcal{M}_{\alpha} & = & \left\{Q \in M \text{ s.t. } \int_{\mathbb{R}^r}{dQ(x)}=1, \int_{\mathbb{R}^r}{g(x)dQ(x)}=m(\alpha) \right\}
\label{eqn:SetMalpha}
\end{eqnarray}
for $\lambda\in(0,1)$ the proportion of the parametric component, $\theta\in\Theta\subset\mathbb{R}^{d}$ a set of parameters defining the parametric component, $\alpha\in\mathcal{A}\subset\mathbb{R}^{s}$ is the constraints parameter vector and finally $m(\alpha)=(m_1(\alpha),\cdots,m_{\ell}(\alpha))$ is a vector-valued function determining the value of the constraints.
\end{definition}
The identifiability of the model was not questioned in the context of Section \ref{sec:GeneralModelDef} because it suffices that the sets $\mathcal{M}_{\alpha}$ are disjoint (the function $m(\alpha)$ is one-to-one). However, in the context of this semiparametric mixture model, identifiability cannot be achieved only by supposing that the sets $\mathcal{M}_{\alpha}$ are disjoint. 
\begin{definition} \label{def:identifiabilitySemipParaMom}
We say that the two-component semiparametric mixture model subject to linear constraints is identifiable if it verifies the following assertion. If
\begin{equation}
\lambda P_1(.|\theta) + (1-\lambda)P_0 = \tilde{\lambda} P_1(.|\tilde{\theta}) + (1-\tilde{\lambda}) \tilde{P}_0,\quad \text{with } P_0\in\mathcal{M}_{\alpha}, \tilde{P}_0\in\mathcal{M}_{\tilde{\alpha}}, 
\label{eqn:IdenitifiabilityDefEq}
\end{equation}
then $\lambda = \tilde{\lambda},\theta = \tilde{\theta}$ and $P_0=\tilde{P}_0$ (and hence $\alpha=\tilde{\alpha}$).
\end{definition}
This is the same identifiability concept considered in \cite{Bordes06b} where the authors exploited their symmetry assumption over $\mathbb{P}_0$ and built a system of moments equations. They proved that if $P_1$ is also symmetric, then equation (\ref{eqn:IdenitifiabilityDefEq}) has two solutions, otherwise it has three solutions. Their idea appears here in a natural way in order to prove the identifiability of our semiparametric mixture model (\ref{eqn:SetMalpha}).
\begin{proposition}
\label{prop:identifiabilityMixture}
For a given mixture distribution $P_T = P(.| \phi^*)$, suppose that the system of equations:
\[\frac{1}{1-\lambda} m^* - \frac{\lambda}{1-\lambda}m_1(\theta) = m_0(\alpha)\]
where $m^*=\int g(x)dP_T(x)$ and $m_1(\theta) = \int g(x) dP_1(x|\theta)$, has a unique solution $(\lambda^*,\theta^*,\alpha^*)$. Then, equation (\ref{eqn:IdenitifiabilityDefEq}) has a unique solution, i.e. $\lambda = \tilde{\lambda},\theta = \tilde{\theta}$ and $P_0=\tilde{P}_0$, and the semiparametric mixture model $P_T = P(.| \phi^*)$ is identifiable.
\end{proposition}
The proof is differed to Appendix \ref{AppendSemiPara:PropIdenitifiability}.
%Example......
\begin{example}[Semiparametric two-component Gaussian mixture]
Suppose that $P_1(.|\theta)$ is a Gaussian model $\mathcal{N}(\mu_1,1)$. Suppose also that the set of constraints is defined as follows:
\[\mathcal{M}_{\mu_0^*} = \left\{f_0 \text{ s.t. } \int{f_0(x)dx}=1,\quad\int_{\mathbb{R}}{xf_0(x)dx} = \mu_0^*,\quad \int_{\mathbb{R}}{x^2f_0(x)dx} = 1+{\mu_0^*}^2\right\}.\]
We would like to study the identifiability of the two-component semiparametric Gaussian mixture whose unknown component $P_0$ shares the first two moments with the Gaussian distribution $\mathcal{N}(\mu_0^*,1)$ for a known $\mu_0^*$. Using Proposition \ref{prop:identifiabilityMixture}, it suffices to study the system of equations 
\begin{eqnarray*}
\frac{1}{1-\lambda} \int{xf(x|\mu_1^*,\mu_0^*,\lambda^*)dx} - \frac{\lambda}{1-\lambda} \int{xf_1(x|\mu_1)dx} & = & \mu_0^* \\
\frac{1}{1-\lambda} \int{x^2f(x|\mu_1^*,\mu_0^*,\lambda^*)dx} - \frac{\lambda}{1-\lambda} \int{x^2f_1(x|\mu_1)dx} & = & 1+{\mu_0^*}^2.
\end{eqnarray*}
Recall that $\int{xf(x|\mu_1^*,\mu_0^*,\lambda^*)dx} = \lambda^*\mu_1^*+(1-\lambda^*)\mu_0^*$ and $\int{x^2f(x|\mu_1^*,\mu_0^*,\lambda^*)dx}=1+\lambda^*{\mu_1^*}^2+(1-\lambda^*){\mu_0^*}^2$. The first equation in the previous system entails that:
\begin{equation}
\lambda\mu_1 - \lambda\mu_0^* = \lambda^*\mu_1^* -\lambda^*\mu_0^*.
\label{eqn:GaussExampleCond1}
\end{equation}
The second equation gives:
\begin{equation}
\lambda^*(1+{\mu_1^*}^2)-\lambda^*(1+{\mu_0^*}^2) = \lambda\mu_1^2 - \lambda{\mu_0^*}^2
\label{eqn:GaussExampleCond2}
\end{equation}
The nonlinear system of equations (\ref{eqn:GaussExampleCond1}, \ref{eqn:GaussExampleCond2}) has a solution for $\mu_1=\mu_1^*, \lambda=\lambda^*$. Suppose by contradiction that $\mu_1\neq\mu_1^*$ and check if there are other solutions. The system (\ref{eqn:GaussExampleCond1}, \ref{eqn:GaussExampleCond2}) implies:
\[
\lambda = \frac{\lambda^*(1+{\mu_1^*}^2)-\lambda^*(1+{\mu_0^*}^2)}{(\mu_1-\mu_0^*)(\mu_1+\mu_0^*)} = \frac{\lambda^*\mu_1^* - \lambda^*\mu_0^*}{\mu_1-\mu_0^*}
\]
This entails that 
\[\mu_1 + \mu_0^* = \frac{\lambda^*\left[(1+{\mu_1^*}^2)-(1+{\mu_0^*}^2)\right]}{\lambda^*\left[\mu_1^* - \mu_0^*\right]}=\mu_1^*+\mu_0^*\] 
Hence, $\mu_1=\mu_1^*$ which contradicts what we have assumed. Thus $\mu_1=\mu_1^*, \lambda=\lambda^*$ is the only solution. We conclude that if $\mu_0^*$ is known and that we impose two moments constraints over $f_0$, then the semiparametric two-component Gaussian mixture model is identifiable.\\
Notice that imposing only one condition on the first moment is not sufficient since any value of $\lambda\in(0,1)$ would produce a corresponding solution for $\mu_1$ in equation (\ref{eqn:GaussExampleCond1}). We therefore are in need for the second constraint. Notice also that if $\lambda=\lambda^*$, then $\mu_1=\mu_1^*$. This means that, by continuity of the equation over $(\lambda,\mu_1)$, if $\lambda$ is initialized in a close neighborhood of $\lambda^*$, then $\mu_1$ would be estimated near $\mu_1^*$. This may represent a remedy if we could not impose but one moment constraint.

\end{example} 
%%%%%%%%%%%%%%%%%%%%%%%%%%%%%%%%%%%%%%%%%%%%%%%%%%%%%%%%%%%%%%%%%%%%%%%%%%%%%%%
%%%%%%%%%%%%%%%%%%%%%%%%%%%%%%%%%%%%%%%%%%%%%%%%%%%%%%%%%%%%%%%%%%%%%%%%%%%%%%
\subsection{An algorithm for the estimation of the semiparametric mixture model}\label{subsec:procIntrod}
We have seen in paragraph \ref{subsec:DualTechRes} that it is possible to use $\varphi-$divergences to estimate a semiparametric model as long as the constraints apply over $P(.|\phi)$, i.e. the whole mixture. In our case, the constraints apply only on a component of the mixture. It is thus reasonable to consider a "model" expressed through $P_0$ instead of $P$. We have:
\[P_0 = \frac{1}{1-\lambda} P(.|\phi) - \frac{\lambda}{1-\lambda} P_1(.|\theta).\]
Denote $P_T=P(.|\phi^*)$ with $\phi^*=(\lambda^*,\theta^*,\alpha^*)$ the distribution which generates the observed data. Denote also $P_0^*$ to the true semiparametric component of the mixture $P_T$. The only information we hold about $P_0^*$ is that it belongs to a set $\mathcal{M}_{\alpha^*}$ for some (possibly unknown) $\alpha^*\in\mathcal{A}$. Besides, it verifies:
\begin{equation}
P_0^* = \frac{1}{1-\lambda^*} P_T - \frac{\lambda^*}{1-\lambda^*} P_1(.|\theta^*).
\label{eqn:TrueP0model}
\end{equation}
We would like to retrieve the value of the vector $\phi^*=(\lambda^*,\theta^*,\alpha^*)$ provided a sample $X_1,\cdots,X_n$ drawn from $P_T$ and that $P_0^*\in\cup_{\alpha}\mathcal{M}_{\alpha}$. Consider the set of probability measures:
\begin{equation}
\mathcal{N} = \left\{Q \in M^+, \; \text{s.t.} \;Q=\frac{1}{1-\lambda} P_T - \frac{\lambda}{1-\lambda} P_1(.|\theta), \quad \lambda\in(0,1),\theta\in\Theta\right\}.
\label{eqn:SetNnewModel}
\end{equation}
Notice that $P_0^*$ belongs to this set for $\lambda=\lambda^*$ and $\theta=\theta^*$. On the other hand, $P_0^*$ is supposed, for simplicity, to belong to the union $\cup_{\alpha\in\mathcal{A}}\mathcal{M}_{\alpha}$. We may now write,
\[P_0^*\in \mathcal{N} \bigcap \cup_{\alpha\in\mathcal{A}}\mathcal{M}_{\alpha}.\]
If we suppose now that the intersection $\mathcal{N} \bigcap \cup_{\alpha\in\mathcal{A}}\mathcal{M}_{\alpha}$ contains only one element which would be a fortiori $P_0^*$, then it is very reasonable to consider an estimation procedure by calculating some "distance" between the two sets $\mathcal{N}$ and $\cup_{\alpha\in\mathcal{A}}\mathcal{M}_{\alpha}$. Such distance can be measured using a $\varphi-$divergence by (see Definition \ref{def:phiDistance}):
\begin{equation}
D_{\varphi}(\mathcal{M},\mathcal{N}) = \inf_{Q\in\mathcal{N}}\inf_{P_0\in\mathcal{M}}D_{\varphi}(P_0,Q).
\label{eqn:DistanceTwoMeasureSets}
\end{equation}
We may reparametrize this distance using the definition of $\mathcal{N}$. Indeed,
\begin{eqnarray}
D_{\varphi}(\cup_{\alpha}\mathcal{M}_{\alpha},\mathcal{N}) & = & \inf_{Q\in\mathcal{N}}\inf_{P_0\in\cup_{\alpha}\mathcal{M}_{\alpha}}D_{\varphi}(P_0,Q) \nonumber\\
 & = & \inf_{\lambda,\theta}\inf_{\alpha, P_0\in\mathcal{M}_{\alpha}}D_{\varphi}\left(P_0,\frac{1}{1-\lambda} P_T - \frac{\lambda}{1-\lambda} P_1(.|\theta)\right).
\label{eqn:DistanceModelConstraint}
\end{eqnarray}
Here, the optimization is carried over the couples $(\lambda,\theta)$ for which the measure $\frac{1}{1-\lambda} P_T - \frac{\lambda}{1-\lambda} P_1(.|\theta)$ is positive (so it is a probability measure). Define the set
\[\Phi^{++} = \left\{(\lambda,\theta,\alpha)\in\Phi, \; \frac{1}{1-\lambda} P_T - \frac{\lambda}{1-\lambda} P_1(.|\theta)\; \text{ is a probability measure}\right\}.\]
If we still have $P_0^*$ as the only (probability) measure which belongs to both $\mathcal{N}$ and $\cup_{\alpha}\mathcal{M}_{\alpha}$, then the argument of the infimum in (\ref{eqn:DistanceModelConstraint}) is none other than $(\lambda^*,\theta^*,\alpha^*)$, i.e.
\begin{equation}
(\lambda^*,\theta^*,\alpha^*) = \arginf_{\lambda,\theta,\alpha\in\Phi^{++}}\inf_{P_0\in\mathcal{M}_{\alpha}}D_{\varphi}\left(P_0,\frac{1}{1-\lambda} P(.|\phi^*) - \frac{\lambda}{1-\lambda} P_1(.|\theta)\right).
\label{eqn:MomentTrueEstimProc}
\end{equation}
It is important to notice that if $P_0^*\notin\cup\mathcal{M}_{\alpha}$, then the procedure still makes sense. Indeed, we are searching for the best measure of the form $\frac{1}{1-\lambda}P_T - \frac{\lambda}{1-\lambda}P_1(.|\theta)$ which verifies the constraints from a point of view of the $\varphi-$divergence.
%%%%%%%%%%%%%%%%%%%%%%%%%%%%%%%%%%%%%%%%%%%%%%%%%%%%%%%%%%%%%%%%%%%%%%%%%%%%%%%%%%%%%%
%%%%%%%%%%%%%%%%%%%%%%%%%%%%%%%%%%%%%%%%%%%%%%%%%%%%%%%%%%%%%%%%%%%%%%%%%%%%%%%%%%%%%%
\subsection{The algorithm in practice : Estimation using the duality technique and plug-in estimate} %% continue
The Fenchel-Legendre duality permits to transform the problem of minimizing under linear constraints in a possibly infinite dimensional space into an unconstrained optimization problem in the space of Lagrangian parameters over $\mathbb{R}^{\ell+1}$, where $\ell+1$ is the number of constraints. We will apply the duality result presenter earlier in paragraph \ref{subsec:DualTechRes} on the inner optimization in equation (\ref{eqn:DistanceModelConstraint}). Redefine the function $m$ (see \ref{eqn:SetMalpha}) as $m(\alpha)=(m_0(\alpha),m_1(\alpha),\cdots,m_{\ell}(\alpha))$ where $m_0(\alpha)=1$. We have:
\begin{eqnarray*}
\inf_{Q\in\mathcal{M}_{\alpha}} D_{\varphi}\left(Q,\frac{1}{\lambda-1} P_T - \frac{\lambda}{1-\lambda} P_1(.|\theta)\right) & = & \sup_{\xi\in\mathbb{R}^{l+1}} \xi^t m(\alpha)  - \frac{1}{1-\lambda}\int{\psi\left(\xi^t g(x)\right) \left(dP_T(x) - \lambda dP_1(x|\theta)\right)}\\
 & = & \sup_{\xi\in\mathbb{R}^{l+1}} \xi^t m(\alpha)  - \frac{1}{1-\lambda}\int{\psi\left(\xi^t g(x)\right) dP_T(x)} \\
	&  & + \frac{\lambda}{1-\lambda} \int{\psi\left(\xi^t g(x)\right) dP_1(x|\theta)}.
\end{eqnarray*}
Inserting this result in (\ref{eqn:MomentTrueEstimProc}) gives that:
\begin{eqnarray*}
(\lambda^*,\theta^*,\alpha^*)  & = &  \arginf_{\phi\in\Phi^{++}}\inf_{Q\in\mathcal{M}_{\alpha}} D_{\varphi}\left(Q,\frac{1}{\lambda-1} P_T - \frac{\lambda}{1-\lambda} P_1(.|\theta)\right) \\
& = & \arginf_{\phi\in\Phi^{++}}\sup_{\xi\in\mathbb{R}^{l+1}} \xi^t m(\alpha)  - \frac{1}{1-\lambda}\int{\psi\left(\xi^t g(x)\right) dP_T(x)} \\
	&  & + \frac{\lambda}{1-\lambda} \int{\psi\left(\xi^t g(x)\right) dP_1(x|\theta)}.
\end{eqnarray*}
The right hand side can be estimated on the basis of an $n-$sample drawn from $P_T$, say $X_1,\cdots,X_n$, by a simple plug-in of the empirical measure $P_n$. Before we proceed to this final step, it is interesting to notice that the characterization of the set $\Phi^++$ is difficult. We will show that we can carry out the optimization over the whole set $\Phi$. Let $H(\phi,\xi)$ be objective function in the previous optimization problem. In other words
\[H(\phi,\xi) = \xi^t m(\alpha)  - \frac{1}{1-\lambda}\int{\psi\left(\xi^t g(x)\right) dP_T(x)} + \frac{\lambda}{1-\lambda} \int{\psi\left(\xi^t g(x)\right) dP_1(x|\theta)}.\]
Define first the set 
\[\Phi^+ = \left\{\phi\in\Phi \text{ s.t. } \xi\mapsto H(\phi,\xi) \text{ is strictly concave}\right\}\]
Notice that for any $\phi\in\Phi\setminus\Phi^+$, we have $\sup_\xi H(\phi,\xi)=\infty$, thus these vectors do not participate in the calculus of the infimum over $\phi$ afterwords. Let us concentrate now on the vectors inside the set $\Phi^+$. Let $\xi(\phi)$ be the supremum of $H(\phi,\xi)$ over $\xi$. This supremum exists and is uniquely defined because $\xi\mapsto H(\phi,\xi)$ is strictly concave. Notice that for any $\phi\in\Phi^+$, we have:
\begin{eqnarray*}
H(\phi,\xi(\phi)) & = & \sup_{\xi} H(\phi,\xi) \\
   & \geq & H(\phi,0) \\
	& \geq & 0.
\end{eqnarray*}
The last line is due to the fact that $\psi(0)=0$.\\
On the other hand, we easily have $\phi^* \in \Phi^{++}\subset\Phi^+$. Moreover, at $\phi^*$ we have a duality attainment in the sens that:
\begin{eqnarray*}
\inf_{P_0\in\mathcal{M}_{\alpha^T}}D_{\varphi}\left(P_0,\frac{1}{1-\lambda^T}P_T - \frac{\lambda^T}{1-\lambda^T}P_1(.|\theta^T)\right) & = & \sup_{\xi} H(\phi^T,\xi)\\
0 & = & H(\phi^T,\xi(\phi^T)).
\end{eqnarray*}
The final line comes from our hypothesis that the probability measure $P_0^T = \frac{1}{1-\lambda^T}P_T - \frac{\lambda^T}{1-\lambda^T}P_1(.|\theta^T)$ verifies the constraints for $\alpha=\alpha^t$, i.e. $P_0^T\in\mathcal{M}_{\alpha^*}$. We thus have proved that $\phi^*$ is a global infimum of $\phi\mapsto H(\phi,\xi(\phi))$ over $\Phi^+$. Since $H(\phi,\xi)=\infty$ for any $\phi\in\Phi\setminus\Phi^+$, we can write that:
\begin{eqnarray}
\phi^* & = & \inf_{\phi\in\Phi}\sup_{\xi\in\mathbb{R}^{\ell+1}} H(\phi,\xi) \nonumber\\
 & = & \inf_{\phi\in\Phi}\sup_{\xi\in\mathbb{R}^{\ell+1}}  \xi^t m(\alpha)  - \frac{1}{1-\lambda}\int{\psi\left(\xi^t g(x)\right) dP_T(x)} \nonumber\\
	&  & + \frac{\lambda}{1-\lambda} \int{\psi\left(\xi^t g(x)\right) dP_1(x|\theta)}.
\label{eqn:MomentExactEstimProc}
\end{eqnarray}
We may now state our estimator by a simple plug-in of the empirical measure in the previous display. 
\begin{eqnarray}
(\hat{\lambda}, \hat{\theta},\hat{\alpha}) & = & \arginf_{(\lambda, \theta,\alpha)\in\Phi} \sup_{\xi\in\mathbb{R}^{l+1}} \xi^t m(\alpha)  - \frac{1}{1-\lambda}\frac{1}{n}\sum_{i=1}^n{\psi\left(\xi^t g(X_i)\right)} \nonumber\\
 &  & \qquad \qquad \qquad + \frac{\lambda}{1-\lambda} \int{\psi\left(\xi^t g(x)\right) dP_1(x|\theta)}.
\label{eqn:MomentEstimProc}
\end{eqnarray}
This is a feasible procedure in the sens that we only need the data, the set of constraints, the model of the parametric component and finally an optimization algorithm (CG, BFGS, SANA, Nelder-Mead,etc) in order to produce an estimate of the true vector of parameters $\phi^*$.
\begin{remark}
It is straightforward that the set $\Phi^+$ of effective parameters in the optimization problem (\ref{eqn:MomentExactEstimProc}) is not empty since it contains $\phi^*$. For the case of the optimization problem (\ref{eqn:MomentEstimProc}), the effective set of parameters $\Phi_n^+$ defined by:
\[\Phi_n^+ = \{\phi\in\Phi, \text{ s.t. }\; H_n(\phi,\xi) \text{ is strictly concave}\},\]
where:
\[H_n(\phi,\xi) = \xi^t m(\alpha)  - \frac{1}{1-\lambda}\frac{1}{n}\sum_{i=1}^n{\psi\left(\xi^t g(X_i)\right)} + \frac{\lambda}{1-\lambda} \int{\psi\left(\xi^t g(x)\right) dP_1(x|\theta)},\]
does not necessarily contain $\phi^*$ for a finite $n$. It is important to show that it is not void. Indeed, define the set $\Phi_n^{++}$ similarly to $\Phi^{++}$
\[\Phi_n^{++} = \{\phi\in\Phi, \text{ s.t. } \frac{1}{1-\lambda}P_T - \frac{\lambda}{1-\lambda}P_1(.|\theta)\text{ is a probability measure}\}.\]
Clearly $\Phi_n^{++}\subset\Phi_n^+$. Besides, for a fixed $\theta$, as $\lambda$ approaches from zero, the influence of $P_1$ disappears since $\frac{\lambda}{1-\lambda}$ becomes aslo close to zero. Thus the measure $\frac{1}{1-\lambda}P_T - \frac{\lambda}{1-\lambda}P_1(.|\theta)$ becomes automatically positive and thus a probability measure. \\
In the general case scenario when $\lambda$ is supposed to be lower-bounded far from zero, the previous argument may not hold, and a more profound study of the objective function $H(\phi,\xi)$ is needed. In case, $\xi\mapsto H(\phi,\xi)$ is differentiable, we can use Sylvester's rule to check if the Hessien matrix is definite negative.
\end{remark}
%%%%%%%%%% Example
\begin{example}[Chi square] 
\label{ex:Chi2LinConstr}
Consider the $\chi^2$ divergence for which $\varphi(t)=(t-1)^2/2$. The Convex conjugate of $\varphi$ is given by $\psi(t)=t^2/2+t$. For $(\lambda,\theta,\alpha)\in\Phi$, we have:
\begin{multline*}
\inf_{Q\in\mathcal{M}_{\alpha}} D_{\varphi}\left(Q,\frac{1}{\lambda-1} P_T - \frac{\lambda}{1-\lambda} P_1(.|\theta)\right)
 =  \sup_{\xi\in\mathbb{R}^{l+1}} \xi^t m(\alpha)  - \frac{1}{1-\lambda}\int{\left[\frac{1}{2}\left(\xi^t g(x)\right)^2+\xi^t g(x)\right] dP_T(x)} \\
	+ \frac{\lambda}{1-\lambda} \int{\left[\frac{1}{2}\left(\xi^t g(x)\right)^2+\xi^t g(x)\right] dP_1(x|\theta)}.
\end{multline*}
It is interesting to note that the supremum over $\xi$ can be calculated explicitly. Clearly, the optimized function is a polynomial of $\xi$ and thus infinitely differentiable. The Hessian matrix is equal to $-\Omega$ where:
\begin{equation}
\Omega = \int{g(x)g(x)^t\left(\frac{1}{1-\lambda}dP(x)-\frac{\lambda}{1-\lambda}dP_1(x|\theta)\right)}. \\
\label{eqn:OmegaMat}
\end{equation}
If the measure $\frac{1}{1-\lambda}dP-\frac{\lambda}{1-\lambda}dP_1(.|\theta)$ is positive, then $\Omega$ is symmetric definite positive (s.d.p) and the Hessian matrix is symmetric definite negative. Consequently, the supremum over $\xi$ is ensured to exist. If it is a signed measure, then the supremum might be infinity. We may now write:
\[
\xi(\phi) = \Omega^{-1}\left(m(\alpha) - \int{g(x)\left(\frac{1}{1-\lambda}dP(x)-\frac{\lambda}{1-\lambda}dP_1(x|\theta)\right)}\right), \quad \text{if } \Omega \text{ is s.d.p}
\]
For the empirical criterion, we define similarly $\Omega_n$ by
\begin{equation}
\Omega_n = \frac{1}{n}\frac{1}{1-\lambda}\sum_{i=1}^n{g(X_i)g(X_i)^t}-\frac{\lambda}{1-\lambda}\int{g(x)g(x)^tdP_1(x|\theta)}. \\
\label{eqn:OmeganMat}
\end{equation}
The solution to the corresponding supremum over $\xi$ is given by
\[
\xi_n(\phi) = \Omega_n^{-1}\left(m(\alpha) - \frac{1}{n}\frac{1}{1-\lambda}\sum_{i=1}^n{g(X_i)}+\frac{\lambda}{1-\lambda}\int{g(x)dP_1(x|\theta)}\right), \quad \text{if } \Omega_n \text{ is s.d.p}\]
\end{example}
%%%%%%%%%%%%%%%%%%%%%%%%%%%%%%%%%%%%%%%%%%%%%%%%%%%%%%%%%%%%%%%%%%%
%%%%%%%%%%%%%%%%%%%%%%%%%%%%%%%%%%%%%%%%%%%%%%%%%%%%%%%%%%%%%%%%%%%
\subsection{Uniqueness of the solution "under the model"}
By a unique solution we mean that only one measure, which can be written in the form of $\frac{1}{1-\lambda}P_T-\frac{\lambda}{1-\lambda}P_1(.|\theta)$, verifies the constraints with a unique triplet $(\lambda^*,\theta^*,\alpha^*)$. The existence of a unique solution is essential in order to ensure that the procedure (\ref{eqn:MomentTrueEstimProc}) is a reasonable estimation method. We provide next a result ensuring the uniqueness of the solution. The idea is based on the identification of the intersection between the set $\mathcal{N}$ (see \ref{eqn:SetNnewModel}) and the constraints $\mathcal{M}$. The proof is differed to Appendix \ref{AppendSemiPara:Prop1}.
\begin{proposition}
\label{prop:identifiability}
Assume that $P_0^*\in\mathcal{M}=\cup_{\alpha}\mathcal{M}_{\alpha}$. Suppose also that:
\begin{enumerate}
\item the system of equations:
\begin{equation}
\int{g_i(x)\left(dP(x|\phi^*) - \lambda dP_1(x|\theta)\right)} = (1-\lambda)m_i(\alpha), \qquad  i=1,\cdots,\ell
\label{eqn:NlnSys}
\end{equation}
has a unique solution $(\lambda^*,\theta^*,\alpha^*)$;
\item the function $\alpha\mapsto m(\alpha)$ is one-to-one;
\item for any $\theta\in\Theta$ we have :
\[\lim_{\|x\|\rightarrow \infty} \frac{dP_1(x|\theta)}{dP_T(x)} = c,\quad \text{ with } c\in [0,\infty)\setminus\{1\};\]
\item the parametric component is identifiable, i.e. if $P_1(.|\theta) = P_1(.|\theta')\;\; dP_T-$a.e. then $\theta=\theta'$,
\end{enumerate}
then, the intersection $\mathcal{N}\cap\mathcal{M}$ contains a unique measure $P_0^*$, and there exists a unique vector $(\lambda^*,\theta^*,\alpha^*)$ such that $P_T = \lambda^*P_1(.|\theta^*)+(1-\lambda)P_0^*$ where $P_0^*$ is given by (\ref{eqn:TrueP0model}) and belongs to $\mathcal{M}_{\alpha^*}$.
\end{proposition}
There is no general result for a non linear system of equations to have a unique solution; still, it is necessary to ensure that $\ell\geq d+s+1$, otherwise there would be an infinite number of signed measures in the intersection $\mathcal{N} \bigcap \cup_{\alpha\in\mathcal{A}}\mathcal{M}_{\alpha}$.
\begin{remark}
Assumptions 3 and 4 of Proposition \ref{prop:identifiability} are used to prove the identifiability of the "model" $\left(\frac{1}{1-\lambda}P_T - \frac{\lambda}{1-\lambda}P_1(.|\theta)\right)_{\lambda,\theta}$. Thus, according to the considered situation we may find simpler ones for particular cases (or even for the general case). Our assumptions remain sufficient but not necessary for the proof. Note also that similar assumption to 3 can be found in the literature on semiparametric mixture models, see Proposition 3 in \cite{Bordes06b}. We can imagine the case of a semiparametric mixture model with a Gaussian parametric component with unknown mean. Besides, the unknown component could have a heavier tail such as an exponential model. Assumption 3 is then fulfilled with $c=0$.
\end{remark}
%%%%%%%%%%%%
% Example
\begin{example}
One of the most popular models in clustering is the Gaussian multivariate mixture (GMM). Suppose that we have two classes. Linear discriminant analysis (LDA) is based on the hypothesis that the covariance matrix of the two classes is the same. Let $X$ be a random variable which takes its values in $\mathbb{R}^2$ and is drawn from a mixture model of two components. In the context of LDA, the model has the form:
\[f(x,y|\lambda,\mu_1,\mu_2,\Sigma) = \lambda f_1(x,y|\mu_1,\Sigma) + (1-\lambda)f_1(x,y|\mu_2,\Sigma),\]
with:
\[f_1(x,y|\mu_1,\Sigma)=\frac{1}{2\pi\sqrt{|\text{det}(\Sigma)|}}\exp\left[-\frac{1}{2}((x,y)^t-\mu_1)^t\Sigma((x,y)^t-\mu_1)\right], \quad \Sigma = \left(\begin{array}{cc}\sigma^2 & \rho \\ \rho & \sigma^2\end{array}\right).\]
We would like to relax the assumption over the second component by keeping the fact that the covariance matrix is the same as the one of the first component. We will start by imposing the very natural constraints on the second component. 
\begin{eqnarray*}
\int{xf_0(x,y)dxdy} & = & \mu_{2,1}, \\
\int{yf_0(x,y)dxdy} & = & \mu_{2,2},\\ 
\int{x^2f_0(x,y)dxdy} & = & \sigma^2, \\
\int{y^2f_0(x,y)dxdy} & = & \sigma^2, \\
\int{xyf_0(x,y)dxdy} & = & \rho + \mu_{2,1}\mu_{2,2} - \mu_{2,1}^2 - \mu_{2,2}^2.
\end{eqnarray*}
These constraints concern only the fact that the covariance matrix $\Sigma$ is the same as the one of the Gaussian component (the parametric one). In order to see whether this set of constraints is sufficient for the existence of a unique measure in the intersection $\mathcal{N}\cap\mathcal{M}$, we need to write the set of equations corresponding to (\ref{eqn:NlnSys}) in Proposition \ref{prop:identifiability}.
\begin{eqnarray*}
\int{x\left[\frac{1}{1-\lambda}f(x,y) - \frac{\lambda}{1-\lambda}f_1(x,y|\mu_1,\sigma,\rho)\right]dxdy} & = & \mu_{2,1}, \\
\int{y\left[\frac{1}{1-\lambda}f(x,y) - \frac{\lambda}{1-\lambda}f_1(x,y|\mu_1,\sigma,\rho)\right]dxdy} & = & \mu_{2,2}, \\ 
\int{x^2\left[\frac{1}{1-\lambda}f(x,y) - \frac{\lambda}{1-\lambda}f_1(x,y|\mu_1,\sigma,\rho)\right]dxdy} & = & \sigma^2, \\
\int{y^2\left[\frac{1}{1-\lambda}f(x,y) - \frac{\lambda}{1-\lambda}f_1(x,y|\mu_1,\sigma,\rho)\right]dxdy} & = & \sigma^2, \\
\int{xy\left[\frac{1}{1-\lambda}f(x,y) - \frac{\lambda}{1-\lambda}f_1(x,y|\mu_1,\sigma,\rho)\right]dxdy} & = & \rho + \mu_{2,1}\mu_{2,2} - \mu_{2,1}^2 - \mu_{2,2}^2,
\end{eqnarray*}
The number of parameters is 7, and we only have 5 equations. In order for the problem to have a unique solution, it is necessary to either add two other constraints or to consider for example $\mu_1 = (\mu_{1,1},\mu_{1,2})$ to be known\footnote{or estimated by another procedure such as $k-$means.}. Other solutions exist, but depend on the prior information. We may imagine an assumption of the form $\mu_{1,1}=a\mu_{1,2}$ and $\mu_{2,1}=b\mu_{2,2}$ for given constants $a$ and $b$.\\
The gain from relaxing the normality assumption on the second component is that we are building a model which is not constrained to a Gaussian form for the second component, but rather to a form which suits the data. The price we pay is the number of relevant constraints which must be at least equal to the number of unknown parameters.
\end{example}

%%%%%%%%%%%%%%%%%%%%%%%%%%%%%%%%%%%%%%%%%%%%%%%%%%%%%%%%%%%%%%%%%%%%%%%%%%%%%%%%%
%
%==============================================================
%%%%%%%%%%%%%%%%%%%%%%%%%%%%%%%%%%%%%%%%%%%%%%%%%%%%%%%%%%%%%%%%%%%%%%%%%%%%
%==============================================================
%
%%%%%%%%%%%%%%%%%%%%%%%%%%%%%%%%%%%%%%%%%%%%%%%%%%%%%%%%%%%%%%%%%%%%%%%%%%%%%%%%%

\section{Asymptotic properties of the new estimator}\label{sec:AsymptotResults}
\subsection{Consistency}
The double optimization procedure defining the estimator $\hat{\phi}$ defined by (\ref{eqn:MomentEstimProc}) does not permit us to use M-estimates methods to prove consistency. In \cite{KeziouThesis} Proposition 3.7 and in \cite{BroniaKeziou09} Proposition 3.4, the authors propose a method which can simply be generalized to any double optimization procedure since the idea of the proof slightly depends on the form of the optimized function. In order to restate this result here and give an exhaustive and a general proof, suppose that our estimator $\hat{\phi}$ is defined through the following double optimization procedure. Let $H$ and $H_n$ be two generic functions such that $H_n(\phi,\xi) \rightarrow H(\phi,\xi)$ in probability for any couple $(\phi,\xi)$. Define $\hat{\phi}$ and $\phi^*$ as follows:
\begin{eqnarray*}
\hat{\phi} & = & \arginf_{\phi} \sup_{\xi} H_n(\phi,\xi);\\
\phi^* & = & \arginf_{\phi} \sup_{\xi} H(\phi,\xi).
\end{eqnarray*}
We adapt the following notation:
\[\xi(\phi) = \argsup_{t} H(\phi,t), \qquad \xi_n(\phi) = \argsup_{t} H_n(\phi,t)\]
The following theorem provides sufficient conditions for consistency of $\hat{\phi}$ towards $\phi^*$. This result will then be applied to the case of our estimator.\\
Assumptions:
\begin{itemize}
\item[A1.] the estimate $\hat{\phi}$ exists (even if it is not unique);
\item[A2.] $\sup_{\xi,\phi} \left|H_n(\phi,\xi) - H(\phi,\xi)\right|$ tends to 0 in probability;
\item[A3.] for any $\phi$, the supremum of $H$ over $\xi$ is unique and isolated, i.e. $\forall \varepsilon>0, \forall \tilde{\xi}$ such that $\|\tilde{\xi}-\xi(\phi)\|>\varepsilon$, then there exists $\eta>0$ such that $H(\phi,\xi(\phi)) - H(\phi,\tilde{\xi})>\eta$;
\item[A4.] the infimum of $\phi\mapsto H(\phi,\xi(\phi))$ is unique and isolated, i.e. $\forall \varepsilon>0, \forall \phi$ such that $\|\phi-\phi^*\|>\varepsilon$, there exists $\eta>0$ such that $H(\phi,\xi(\phi))-H(\phi^*,\xi(\phi^*))>\eta$;
\item[A5.] for any $\phi$ in $\Phi$, function $\xi\mapsto H(\phi,\xi)$ is continuous.
\end{itemize}
In assumption A4, we suppose the existence and uniqueness of $\phi^*$. It does not, however, imply the uniqueness of $\hat{\phi}$. This is not a problem for our consistency result. The vector $\hat{\phi}$ may be any point which verifies the minimum of function $\phi\mapsto\sup_{\xi} H_n(\phi,\xi)$. Our consistency result shows that all vectors verifying the minimum of $\phi\mapsto\sup_{\xi} H_n(\phi,\xi)$ converge to the unique vector $\phi^*$. We also prove an asymptotic normality result which shows that even if $\hat{\phi}$ is not unique, all possible values should be in a neighborhood of radius $\mathcal{O}(n^{-1/2})$ centered at $\phi^*$.\\
The following lemma establishes a uniform convergence result for the argument of the supremum over $\xi$ of function $H_n(\phi,\xi)$ towards the one of function $H(\phi,\xi)$. It constitutes a first step towards the proof of convergence of $\hat{\phi}$ towards $\phi^*$. The proof is differed to Appendix \ref{AppendSemiPara:Lem1}.
%%%%%%%%%%%%%%%%%%%%%%%%%%%
\begin{lemma}
\label{lem:SupXiPhiDiff}
Assume A2 and A3 are verified, then 
\[\sup_{\phi}\|\xi_n(\phi) - \xi(\phi)\| \rightarrow 0 ,\qquad \textit{in probability.}\] 
\end{lemma}
\noindent We now state our consistency theorem. The proof is differed to Appendix \ref{AppendSemiPara:Theo1}.
%%%%%%%%%%%%%%%%%%%%%%%%%
\begin{theorem}
\label{theo:MainTheorem}
Let $\xi(\phi)$ be the argument of the supremum of $\xi\mapsto H(\phi,\xi)$ for a fixed $\phi$. Assume that A1-A5 are verified, then $\hat{\phi}$ tends to $\phi^*$ in probability.
\end{theorem}
%%%%%%%%%%%%%%%%%%%%%%%%%%%%%
\noindent Let's now go back to our optimization problem (\ref{eqn:MomentEstimProc}) in order to simplify the previous assumptions. First of all, we need to specify functions $H$ and $H_n$. Define the function $h$ as follows. Let $\phi = (\lambda,\theta,\alpha)$,
\[h(\phi,\xi, z) = \xi^t m(\alpha)  - \frac{1}{1-\lambda}\psi\left(\xi^t g(z)\right) + \frac{\lambda}{1-\lambda} \int{\psi\left(\xi^t g(x)\right) dP_1(x|\theta)}.\]
The functions $H$ and $H_n$ can now be defined through $h$ by
\[H(\phi,\xi) = P_T h(\phi,\xi,.),\qquad H_n(\phi,\xi) = P_n h(\phi,\xi,.).\]
%In example \ref{ex:Chi2LinConstr}, we considered the the case of the Pearson's $\chi^2$. The supremum is infinity whenever the matrix $\Omega$ is s.d.n. It is thus interesting to define the \emph{effective} set of parameters. Define the set $\Phi^+$ by
%\[\Phi^+ = \{\phi\in\Phi \text{ s.t. } \xi\mapsto H(\phi,\xi) \text{ is strictly concave}\}\]
%Outside the set $\Phi^+$, function $\xi\mapsto H(\phi,\xi)$ is not upper bounded. 
\begin{theorem}
\label{theo:MainTheoremMomConstr}
Assume that A1, A4 and A5 are verified for $\Phi$ replaced by $\Phi^+$ (defined earlier by). Suppose also that 
\begin{equation}
\sup_{\xi\in\mathbb{R}^{\ell+1}}\left|\int{\psi\left(\xi^t g(x)\right)dP_T(x)} - \frac{1}{n}\sum_{i=1}^n{\psi\left(\xi^t g(X_i)\right)}\right| \xrightarrow[\mathbb{P}]{n\rightarrow\infty} 0,
\label{eqn:AssumptionB2}
\end{equation} 
then the estimator defined by (\ref{eqn:MomentEstimProc}) is consistent.
\end{theorem}

\noindent The proof is differed to Appendix \ref{AppendSemiPara:Theo1}. Assumption A5 could be handled using Lebesgue's continuity theorem if one finds a $P_T-$integrable function $\tilde{h}$ such that $|\psi\left(\xi^t g(z)\right)|\leq \tilde{h}(z)$. This is, however, not possible in general unless we restrain $\xi$ to a compact set. Otherwise, we need to verify this assumption according the situation we have in hand; see example \ref{ex:Chi2Consistency} below for more details. The uniform limit (\ref{eqn:AssumptionB2}) can be treated according to the divergence and the constraints which we would like to impose. A general method is to prove that the class of functions $\{x\mapsto \psi\left(\xi^t g(x)\right),\xi\in\mathbb{R}^{\ell+1}\}$ is a Glivenko-Cantelli class of functions, see \cite{Vaart} Chap. 19 Section 2 and the examples therein for some possibilities.
%existence of $\hat{\phi}$ and continuity of $\xi(\phi)$.
\begin{remark}
\label{rem:PhiPlusJacobMat}
Under suitable differentiability assumptions, the set $\Phi^+$ defined earlier can be rewritten as
\[\Phi^+ = \Phi\cap \left\{\phi: \quad J_{H(\phi,.)} \text{ is definite negative}\right\},\]
where $J_{H(\phi,.)}$ is the Hessian matrix of function $\xi\mapsto H(\phi,\xi)$ and is given by
\begin{equation}
J_{H(\phi,.)} = -\int{g(x)g(x)^t\psi''(\xi^tg(x)) \left(\frac{1}{1-\lambda}dP_T - \frac{\lambda}{1-\lambda}dP_1\right)(x)}.
\label{eqn:JacobHMomentConstr}
\end{equation}
\end{remark}
\noindent The problem with using the set $\Phi^+$ is that if we take a point $\phi$ in the interior of $\Phi$, there is no guarantee that it would be an interior point of $\Phi^+$. This will impose more difficulties in the proof of the asymptotic normality. We prove in the next proposition that this is however true for $\phi^*$. Besides, the set $\Phi^+$ is open as soon as $\text{int}{\Phi}$ is not void. The proof is differed to Appendix \ref{AppendSemiPara:Prop2}.
%%Proposition
\begin{proposition}
\label{prop:ContinDiffxiMom}
Assume that the function $\xi\mapsto H(\phi,\xi)$ is of class $\mathcal{C}^2$ for any $\phi\in\Phi^+$. Suppose that $\phi^*$ is an interior point of $\Phi$, then there exists a neighborhood $\mathcal{V}$ of $\phi^*$ such that for any $\phi\in\mathcal{V}$, $J_{H(\phi,.)}$ is definite negative and thus $\xi(\phi)$ exists and is finite. Moreover, function $\phi\mapsto\xi(\phi)$ is continuously differentiable on $\mathcal{V}$.
\end{proposition}
\begin{corollary}
\label{cor:GlivenkoCantelliClass}
Assume that the function $\xi\mapsto H(\phi,\xi)$ is of class $\mathcal{C}^2$ for any $\phi\in\Phi^+$. If $\Phi$ is bounded, then there exists a compact neighborhood $\bar{\mathcal{V}}$ of $\phi^*$ such that $\xi(\bar{\mathcal{V}})$ is bounded and $\{x\mapsto \psi\left(\xi^t g(x)\right),\xi\in\xi(\bar{\mathcal{V}})\}$ is a Glivenko-Cantelli class of functions.
\end{corollary}
\begin{proof}
The first part of the corollary is an immediate result of Proposition \ref{prop:ContinDiffxiMom} and the continuity of function $\phi\mapsto\xi(\phi)$ over $\text{int}(\Phi^+)$. The implicit functions theorem permits to conclude that $\xi(\phi)$ is continuously differentiable over $\text{int}(\Phi^+)$. The second part is an immediate result of Example 19.7 page 271 from \cite{Vaart}.
\end{proof}
\noindent This corollary suggests that in order to prove the consistency of $\hat{\phi}$, it suffices to restrict the values of $\phi$ on $\Phi^+$ and the values of $\xi$ on $\xi(\text{int}(\Phi^+))$ in the definition of $\hat{\phi}$ (\ref{eqn:MomentEstimProc}). Besides, since $\{x\mapsto \psi\left(\xi^t g(x)\right),\xi\in\xi(\text{int}(\Phi^+))\}$ is a Glivenko-Cantelli class of functions, the uniform limit (\ref{eqn:AssumptionB2}) is verified by the Glivenko-Cantelli theorem.

\begin{remark}
There is a great difference between the set $\Phi^+$ where $J_H$ is s.d.n. (see Remark \ref{rem:PhiPlusJacobMat}) and the set $\Phi^{++}$ where only $\frac{1}{1-\lambda}dP_T - \frac{\lambda}{1-\lambda}dP_1$ is a probability measure. Indeed, there is a strict inclusion in the sense that if $\frac{1}{1-\lambda}dP_T - \frac{\lambda}{1-\lambda}dP_1$ is a probability measure, then $J_H$ is s.d.n., but the inverse does not hold. Figure \ref{fig:PosSetVsPosDefSet} shows this difference. Furthermore, it is clearly simpler to check for a vector $\phi$ if the matrix $J_H$ is s.d.n. It suffices to calculate the integral\footnote{If function $g$ is a polynomial, i.e. moment constraints, then the integral is a mere subtractions between the moments of $P_T$ and the ones of $P_1$.} (even numerically) and then use some rule such as Sylvester's rule to check if it is definite negative, see the example below. However, in order to check if the measure $\frac{1}{1-\lambda}dP_T - \frac{\lambda}{1-\lambda}dP_1$ is positive, we need to verify it on all $\mathbb{R}^r$.
\begin{figure}[h]
\centering
\includegraphics[scale=0.4]{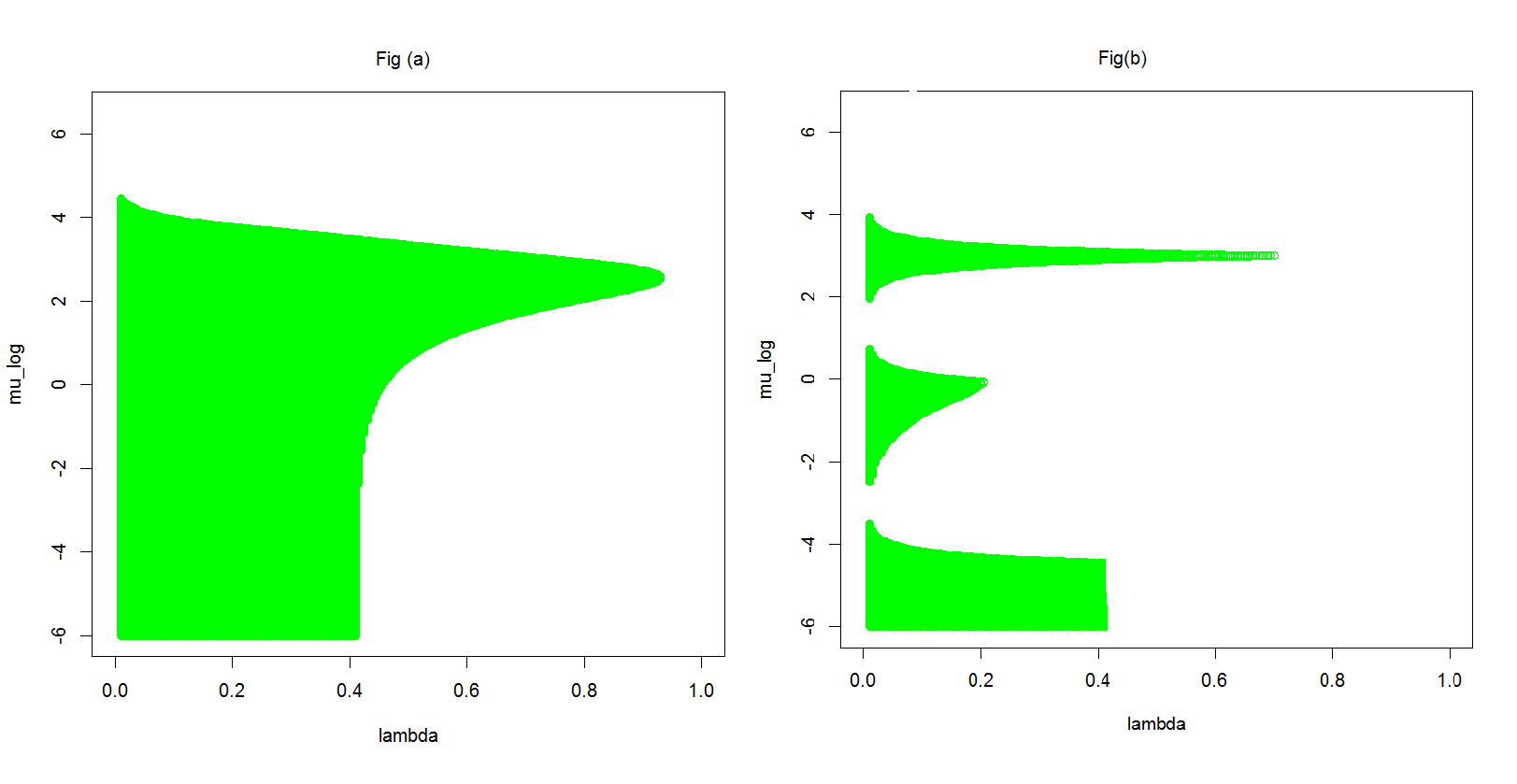}
\caption{Differences between the set where $\frac{1}{1-\lambda}dP_T - \frac{\lambda}{1-\lambda}dP_1$ is positive (Fig (b)) and the set $\Phi^+$ (Fig (a)) in a Weibull--Lognormal mixture.}
\label{fig:PosSetVsPosDefSet}
\end{figure}
\end{remark}
\begin{remark}
The previous remark shows the interest of adapting a methodology based on the larger set $\Phi^+$. We have a larger and better space to search inside for the triplet $(\lambda^*,\theta^*,\alpha^*)$. For example, in Figure \ref{fig:PosSetVsPosDefSet}, the optimization algorithm which tries to solve (\ref{eqn:MomentEstimProc}) gets stuck if we only search in the set of parameters for which $\frac{1}{1-\lambda}dP_T - \frac{\lambda}{1-\lambda}dP_1$ is a probability measure. This does not happen if we search in the set $\Phi^+$. Moreover, even if the algorithm returns a triplet $(\hat{\lambda},\hat{\theta},\hat{\alpha})$ for which the semiparametric component $P_0=\frac{1}{1-\lambda}dP_T - \frac{\lambda}{1-\lambda}dP_1$ is not a probability measure, it should not mean that the procedure failed. This is because we are looking for the parameters and not to estimate $P_0$. Besides, it is still possible to threshold the negative values from the density and then regularize in order to integrate to one. 
\end{remark}
\begin{example}[$\chi^2$ case]
\label{ex:Chi2Consistency}
Consider the case of a two-component semiparapetric mixture model where $P_0$ is defined through its first three moments. In other words, the set of constraints $\mathcal{M}_{\alpha}$ is given by
\[\mathcal{M}_{\alpha} = \left\{Q: \int{dQ(x)}=1,\; \int{xdQ(x)}=m_1(\alpha),\; \int{x^2dQ(x)}=m_2(\alpha),\; \int{x^3dQ(x)}=m_3(\alpha)\right\}.\]
We have already seen in example \ref{ex:Chi2LinConstr} that if $\psi(t)=t^2/2+t$, the Pearson's $\chi^2$ convex conjugate, then the optimization over $\xi$ can be solved and the solution is given by
\[\xi(\phi) = \Omega^{-1}\left(m(\alpha) - \int{g(x)\left(\frac{1}{1-\lambda}dP(x)-\frac{\lambda}{1-\lambda}dP_1(x|\theta)\right)}\right), \text{ for } \phi\in\Phi^+.\]
Let $M_i$ denotes the moment of order $i$ of $P_T$. Denote also $M_i^{(1)}(\theta)$ the moment of order $i$ of the parametric component $P_1(.|\theta)$.
\[M_i = \mathbb{E}_{P_T}[X^i],\qquad M_i^{(1)}(\theta)=\mathbb{E}_{P_1(.|\theta)}[X^i].\]
A simple calculus shows that:
\begin{eqnarray*}
\Omega & = & \int{g(x)g(x)^t\left(\frac{1}{1-\lambda}dP(x)-\frac{\lambda}{1-\lambda}dP_1(x|\theta)\right)} \\
  & = & \left[\frac{1}{1-\lambda}M_{i+j-2} - \frac{\lambda}{1-\lambda}M_{i+j-2}^{(1)}(\theta)\right]_{i,j\in\{1,\cdots,4\}}.
\end{eqnarray*}
The solution holds for any $\phi\in\text{int}(\Phi^+)$. Continuity assumption A5 over $\xi\mapsto H(\phi,\xi)$ is simplified here because function $H$ is a polynomial of degree 2. We have:
\begin{multline*}
H(\phi,\xi) = \xi^tm(\alpha) - \left[\frac{1}{2}\xi_1^2+\xi_1+(\xi_1\xi_2+\xi_2)\left(\frac{1}{1-\lambda}M_1-\frac{\lambda}{1-\lambda}M_1^{(1)}(\theta)\right)\right. \\ +(\xi_2^2/2+\xi_1\xi_2+\xi_3)\left(\frac{1}{1-\lambda}M_2-\frac{\lambda}{1-\lambda}M_2^{(1)}(\theta)\right) + (\xi_1\xi_4+\xi_2\xi_3+\xi_4)\left(\frac{1}{1-\lambda}M_3-\frac{\lambda}{1-\lambda}M_3^{(1)}(\theta)\right) \\
+ (\xi_3^2/2+\xi_2\xi_4)\left(\frac{1}{1-\lambda}M_4-\frac{\lambda}{1-\lambda}M_4^{(1)}(\theta)\right) + \xi_3\xi_4\left(\frac{1}{1-\lambda}M_5-\frac{\lambda}{1-\lambda}M_5^{(1)}(\theta)\right) \\
\left. + \xi_4^2/2 \left(\frac{1}{1-\lambda}M_6-\frac{\lambda}{1-\lambda}M_6^{(1)}(\theta)\right)\right]
\end{multline*}
Regularity of function $\phi\mapsto\xi(\phi)$ is directly tied by the regularity of the moments of $P_1(.|\theta)$ with respect to $\theta$. If $M_i^{(1)}$ is continuous with respect to $\theta$ and $m(\alpha)$ is continuous with respect to $\alpha$, then the existence of $\phi^*$ becomes immediate as soon as the set $\Phi$ is compact.\\
If $\phi^*$ is an interior point of $\Phi$, then Proposition \ref{prop:ContinDiffxiMom} and Corollary \ref{cor:GlivenkoCantelliClass} apply. Thus int$(\Phi^+)$ is non void and the class $\{x\mapsto \psi\left(\xi^t g(x)\right),\xi\in\xi(\text{int}(\Phi^+))\}$ is a Glivenko-Cantelli class of functions. Assumption A4 remains specific to the model we consider. \\
The previous calculus shows that our procedure for estimating $\hat{\phi}$ can be done efficiently and the complexity of the calculus does not depend on the dimension of the data. Besides, no numerical integration is needed.
\end{example}
%%%%%%%%%%%%%%%%%%%%%%%%%%%%%%%%%%%%%%%%%%%%%%%%%%%%%%%
% =============
%%%%%%%%%%%%%%%%%%%%%%%%%%%%%%%%%%%%%%%%%%%%%%%%%%%%%%%
\subsection{Asymptotic normality}
We will suppose that the model $P_1(.|\phi)$ has a density $p_1(.|\phi)$ with respect to the Lebesgue measure which is of class $\mathcal{C}^2(\text{int}(\Phi^+))$ and that $\psi$ is $\mathcal{C}^2(\mathbb{R})$. In order to simplify the formula below, we suppose that $\psi'(0)=1$ and $\psi''(0)=1$. These are not restrictive assumptions and can be relaxed. Recall that they are both verified in the class of Cressie-Read functions (\ref{eqn:CressieReadPhi}).\\
Define the following matrices:
\begin{eqnarray}
J_{\phi^*,\xi^*} & = & \resizebox{0.75\textwidth}{!}{ $\left( \frac{1}{(1-\lambda^*)^2}\left[-\mathbb{E}_{P_T}\left[g(X)\right] + \mathbb{E}_{P_1(.|\theta^*)}\left[g(X)\right]\right], \frac{\lambda^*}{1-\lambda^*}\int{g(x)\nabla_{\theta}p_1(x|\theta^*)dx}, \nabla m(\alpha^*) \right)$}; \label{eqn:NormalAsymMomJ1} \\
J_{\xi^*,\xi^*} & = & \mathbb{E}_{P_0^*}\left[g(X)g(X)^t\right]; \label{eqn:NormalAsymMomJ2}\\
\Sigma & = & \left(J_{\phi^*,\xi^*}^t J_{\xi^*,\xi^*} J_{\phi^*,\xi^*}\right)^{-1}; \label{eqn:NormalAsymMomSigma}\\
H & = & \Sigma J_{\phi^*,\xi^*}^t J_{\xi^*,\xi^*}^{-1}; \label{eqn:NormalAsymMomH}\\  
W & = & J_{\xi^*,\xi^*}^{-1} - J_{\xi^*,\xi^*}^{-1} J_{\phi^*,\xi^*} \Sigma J_{\phi^*,\xi^*}^t J_{\xi^*,\xi^*}^{-1}. \label{eqn:NormalAsymMomP}
\end{eqnarray}
Recall the definition of $\Phi^+$ and define similarly the set $\Phi_n^+$
\begin{eqnarray}
\Phi^+ & = & \left\{\phi: \quad \int{g(x)g(x)^t\left(\frac{1}{1-\lambda}dP_T - \frac{\lambda}{1-\lambda}dP_1\right)(x|\theta)}\text{ is s.p.d.}\right\}; \label{eqn:PhiPlus}\\
\Phi_n^+ & = & \left\{\phi: \quad \frac{1}{n}\frac{1}{1-\lambda}\sum_{i=1}^n{g(X_i)g(X_i)^t} - \frac{\lambda}{1-\lambda}\int{g(x)g(x)^tdP_1(x|\theta)} \text{ is s.p.d.} \right\} \label{eqn:PhiPlusn}.
\end{eqnarray}
These two sets are the feasible sets of parameters for the optimization problems (\ref{eqn:MomentTrueEstimProc}) and (\ref{eqn:MomentEstimProc}) respectively. In other words, outside of the set $\Phi^+$, we have $H(\phi,\xi(\phi))=\infty$. Similarly, outside of the set $\Phi_n^+$, we have $H_n(\phi,\xi_n(\phi))=\infty$.
\begin{theorem}
\label{theo:AsymptotNormalMomConstr}
Suppose that:
\begin{enumerate}
\item $\hat{\phi}$ is consistent and $\phi^*\in\text{int}\left(\Phi\right)$;
\item the function $\alpha\mapsto m(\alpha)$ is $\mathcal{C}^2$;
\item $\forall \phi\in B(\phi^*,\tilde{r})$ and any $\xi\in\xi(B(\phi^*,\tilde{r}))$, there exist functions $h_{1,1},h_{1,2}\in L^1(p_1(.|\theta))$ such that $\left\|\psi'\left(\xi^t g(x)\right)g(x)\right\|\leq h_{1,1}(x)$ and $\left\|\psi''\left(\xi^t g(x)\right)g(x)g(x)^t\right\|\leq h_{1,2}(x)$;
\item $\forall \xi\in\xi(B(\phi^*,\tilde{r}))$, there exist functions $h_{2,1},h_{2,2}\in L^1(dx)$ such that $\left\|\psi\left(\xi^t g(x)\right)\nabla_{\theta}p_1(x|\theta)\right\|\leq h_2(x)$ and $\left\|\psi\left(\xi^t g(x)\right)J_{p_1(.|\theta)}\right\|\leq h_2(x)$;
\item for any couple $(\phi,\xi)\in B(\phi^*,\tilde{r})\times\xi(B(\phi^*,\tilde{r}))$, there exists a function $h_3\in L^1(dx)$ such that $\left\|\psi'\left(\xi^t g(x)\right)g(x)\nabla_{\theta}p_1(x|\theta)^t\right\|\leq h_3(x)$;
\item finite second order moment of $g$ under $P_T$, i.e. $\mathbb{E}_{P_T}\left[g_i(X)g_j(X)\right]<\infty$ for $i,j\leq\ell$;
\item matrices $J_{\xi^*,\xi^*}$ and $J_{\phi^*,\xi^*}^t J_{\xi^*,\xi^*} J_{\phi^*,\xi^*}$ are invertible,
\end{enumerate}
then
\[\left(\begin{array}{c}  \sqrt{n}\left(\hat{\phi}-\phi^*\right) \\ \sqrt{n}\xi_n(\hat{\phi})\end{array}\right) \xrightarrow[\mathcal{L}]{} \mathcal{N}\left(0,\frac{1}{(1-\lambda^*)^2}\left(\begin{array}{c}H \\ W\end{array}\right) \text{Var}_{P_T}(g(X)) \left(H^t\quad W^t\right)\right),\]
where $H$ and $P$ are given by formulas (\ref{eqn:NormalAsymMomH}) and (\ref{eqn:NormalAsymMomP}).
\end{theorem}
The proof is differed to Appendix \ref{AppendSemiPara:Theo3}. Assumption 3 entails the differentiability of function $H_n(\xi,\phi)$ up to second order with respect to $\xi$ whatever the value of $\phi$ in a neighborhood of $\phi^*$. Assumption 4 entails the differentiability of function $H_n(\xi,\phi)$ up to second order with respect to $\theta$ in a neighborhood of $\theta^*$ inside $\xi(B(\phi^*,\tilde{r}))$. Finally, assumption 5 implies the cross-differentiability of function $H_n(\xi,\phi)$ with respect to $\xi$ and $\theta$.\\ 
Differentiability assumptions in Theorem \ref{theo:AsymptotNormalMomConstr} can be relaxed in the case of the Pearson's $\chi^2$ since all integrals in functions $H_n$ and $H$ can be calculated. Our result covers the general case and thus we need to ensure differentiability of the integrals using Lebesgue theorems which requires the existence of integrable functions which upperbound the integrands.
%\begin{remark}
%It is important to notice that the variance of the estimator becomes higher as the proportion of the parametric part becomes higher.
%\end{remark}

%%%%%%%%%%%%%%%%%%%%%%%%%%%%%%%%%%%%%%%%%%%%%%%%%%%%%%%%%%%%%%%%%%%%%%%%%%%%%%%%%
%
%==============================================================
%%%%%%%%%%%%%%%%%%%%%%%%%%%%%%%%%%%%%%%%%%%%%%%%%%%%%%%%%%%%%%%%%%%%%%%%%%%%
%==============================================================
%
%%%%%%%%%%%%%%%%%%%%%%%%%%%%%%%%%%%%%%%%%%%%%%%%%%%%%%%%%%%%%%%%%%%%%%%%%%%%%%%%%

\section{Simulation study}\label{sec:SemiParaSimulations}
We perform several simulations in univariate and multivariate situations and show how prior information about the moments of the distribution of the semiparametric component $P_0$ can help us better estimate the set of parameters $(\lambda^*,\theta^*,\alpha^*)$ in regular examples, i.e. the components of the mixture can be clearly distinguished when we plot the probability density function. We also show how our approach permits to estimate even in difficult situations when the proportion of the parametric component is very low; such cases could \emph{not} be estimated using existing methods.\\
Another important problem in existing methods is their quadratic complexity. For example, an EM-type method such as \cite{Robin}'s algorithm or its stochastic version introduced by \cite{BordesStochEM} performs
%needs at each iteration to calculate first a weighted kernel density estimator and calculate it at each observation of the sample. This calculus has a complexity of order $n^2$. We then need to calculate a vector of weights of length $n$, and estimate the proportion by averaging this vector of weights. Finally, we need to estimate the parameters of $P_1$ by maximum likelihood which can be done at the best cases by averaging $n$ terms. This means that such an algorithm needs to do at least 
$n^2+3n$ operations in order to complete a single iteration. An EM-type algorithm for semiparametric mixture models needs in average 100 iterations to converge and may attain 1000 iterations\footnote{This was the case of the Weibull mixture.} for each sample. To conclude, the estimation procedure performs at least $100(n^2+3n)$ operations. In a signal-noise situations where the signal has a very low proportion around $0.05$, we need a greater number of observations say $n=10^5$. Such experiences cannot be performed using an EM-type method such as the one in \cite{Robin} or its stochastic version introduced in \cite{BordesStochEM} unless one has a "super computer". The symmetry method in \cite{Bordes10} shares similar complexity\footnote{we need more than 24 hours to estimate the parameters of one sample with $10^5$ observations.} $\mathcal{O}(n^2)$.
% because one needs to calculate a cumulative density estimator on each observation, and thus a complexity $n^2$. There is then the optimization step which needs at least 100 iterations to converge\footnote{we need more than 24 hours to estimate the parameters of one sample with $10^5$ observations.}.
Last but not least, the EM-type method of \cite{Song} and their $\pi-$maximizing one have the advantage over other methods, because we need only to calculate a kernel density estimator once and for all, then use it at each iteration\footnote{We were able to perform simulations with $n=10^5$ observations but needed about 5 days on an i7 laptop clocked at 2.5 GHz with 8GB of RAM.}. Nevertheless, the method has still a complexity of order $n^2$.\\
Our approach, although has a double optimization procedure, it can be implemented when $g$ is polynomial and $\varphi$ corresponds to the Pearson's $\chi^2$ in a way that it has a linear complexity $n$. First of all, using the $\chi^2$ divergence, the optimization over $\xi$ in (\ref{eqn:MomentEstimProc}) can be calculated directly. On the other hand, all integrals are mere calculus of empirical moments and moments of the parametric part, see Example \ref{ex:Chi2Consistency}. Empirical moments can be calculated once and for all whereas moments of the parametric part can be calculated using direct formulas available for a large class of probability distributions. What remains is the optimization over $\phi$. In the simulations below, our method produced the estimates instantly even for a number of observations of order $10^7$ whereas other existing methods needed from several hours (algorithms of \cite{Song}) to several days (for other algorithms).\\
Because of the very long execution time of existing methods, we restricted the comparison to simulations in regular situations with $n<10^4$. Experiments with greater number of observations were only treated using our method and the methods in \cite{Song} (for $n\leq 10^5$). In all tables presented hereafter, we performed 100 experiments and calculated the average of resulting estimators. We provided also the standard deviation of the 100 experiments in order to get a measure of preference in case the different estimation methods gave close results.\\ 
In our experiments, the datasets were generated by the following mixtures
\begin{itemize}
\item[$\bullet$] A two-component Weibull mixture;
%\item[$\bullet$] A two-component Weibull - Lognormal mixture;
\item[$\bullet$] A two-component Gaussian -- Two-sided Weibull mixture;
\item[$\bullet$] A two-component bivariate Gaussian mixture.
\end{itemize}
We have chosen a variety of values for the parameters especially the proportion. The second model stems from a signal-noise application where the signal is centered at zero whereas the noise is repartitioned at both sides. The third model appears in clustering and is only presented to show how our method performs in multivariate contexts.\\
Simulations were done using the R tool \cite{Rtool}. Optimization was performed using the Nelder-Mead algorithm, see \cite{NelderMead}. For the $\pi-$maximizing algorithm of \cite{Song}, we used the Brent's method because the optimization was carried over one parameter.\\ 
For our procedure, we only used the $\chi^2$ divergence\footnote{We noticed no great difference when using a Hellinger divergence.}, because our method can be implemented efficiently, see Example \ref{ex:Chi2Consistency}. Recall that the optimized function over $\xi$ is not always strictly concave and the Hessian matrix may be definite positive, see remark \ref{rem:PhiPlusJacobMat}. It is thus important to check for each vector $\phi=(\lambda,\theta,\alpha)$ if the Hessian matrix is still definite negative for example using Sylvester's criterion. If it is not, we set the objective function to a value such as $10^2$ which is high enough since the objective function should have its minimum near zero. Besides, since the resulting function $\phi\mapsto H_n(\phi,\xi_n(\phi))$ as a function of $\phi$ is not ensured to be strictly convex, we used 10 random initial feasible points inside the set $\Phi_n^+$ defined by (\ref{eqn:PhiPlusn}). We then ran the Nelder-Mead algorithm and chose the vector of parameters for which the objective function has the lowest value. We applied a similar procedure on the algorithm of \cite{Bordes10} in order to ensure a \emph{good and fair} optimization.
\begin{remark}
In the literature on the stochastic EM algorithm, it is advised that we iterate the algorithm for some time until it reaches a stable state, then continue iterating long enough and average the values obtained in the second part. The trajectories of the algorithm were very erratic especially for the estimation of the proportion. For us, we iterated for the stochastic EM-type algorithm of \cite{BordesStochEM} 5000 times and averaged the 4000 final iterations.
\end{remark}
\begin{remark}
Initialization of both the EM-type algorithm introduced in \cite{Song} and the SEM-type algorithm introduced in \cite{BordesStochEM} was not very important, and we got the same results when the vector of weights was initialized uniformly or in a "good" way. The EM-type method presented in \cite{Robin} was more influenced by such initialization and we used most of the time a good starting points. In the paper \cite{Robin}, the authors mention that their EM-type algorithm has a fixed point with a proportion at 0 or 1. This confirms that a good initialization is needed in order to avoid theses fixed points.
\end{remark}
\begin{remark}
For the methods introduced in \cite{Song}, we need to estimate mixture's distribution using a kernel density estimator. For the data generated from a Weibull mixture and the data generated from a Weibull Lognormal mixture, we used a reciprocal inverse Gaussian kernel density estimator with a window equal to 0.01 according to the simulation study in \cite{Diaa}. For the symmetry method presented in \cite{Bordes10}, we used a triangular kernel which gave better results than the use of a Gaussian kernel.
\end{remark}
\begin{remark}
Matrix inversion was done manually using direct inversion methods, because the function \texttt{solve} in the statistical program R produced errors sometimes because the matrix was highly sensible at some point during the optimization. For matrices of dimension $4\times 4$ and $5\times 5$ we used block matrix inversion, see for example \cite{BlockMat}. The inverse of a $3\times 3$ was calculated using a direct formula.
\end{remark}

%%%%%%%%%%%%%%%%%%%%%%%%%%%%%%%%%%%%%%%%%%%%%%%
\subsection{Data generated from a two-component Weibull mixture modeled by a semiparametric Weibull mixture}
We consider a mixture of two Weibull components with scales $\sigma_1 = 0.5,\sigma_2=1$ and shapes $\nu_1=2,\nu_2=1$ in order to generate the dataset. In the semiparametric mixture model, the parametric component will be "the one to the right", i.e. the component whose true set of parameters is $(\nu_1=2,\sigma_1=0.5)$. We illustrate several values of the proportion $\lambda\in\{0.7,0.3\}$. This constitutes a difficult example for both our method and existing methods such as EM-type methods or the $\pi-$maximizing algorithm of \cite{Song}. We therefore, simulated 10000-samples and fixed both scales during estimation. We estimate the proportion and the shapes of both components. For our method, the variance of the estimator of $\nu_1$ was high and we needed to use 4 moments to reduce it to an acceptable range. Of course, as the number of observations increases, the variance reduces. We, however, avoided greater number of observations because methods such as \cite{Robin} need very long execution time for even one sample. The method presented in \cite{Bordes10} which is based on a symmetry assumption cannot be applied here since the support of the mixture is $\mathbb{R}_+$.\\
The moments of the Weibull distribution are given by
\[\mathbb{E}[X^i] = \sigma^i\Gamma(1+i/\nu),\qquad\forall i\in\mathbb{N}.\]
The results of our method are clearly better than existing methods which practically failed and could not see but one main component with shape in between the two shapes, see Table \ref{tab:3by3ResultsWeibullMoment}. Although our method presents an inconvenient greater variance for $\nu_1$, the Monte-Carlo mean of the hundred experiences is still unbiased. We believe that the use of other types of constraints would have resulted in better results without the need to add one more constraint. 
\begin{table}[ht]
\centering
\begin{tabular}{|c|c|c|c|c|c|c|}
\hline
Nb of observations & $\lambda$ & sd$(\lambda)$ & $\nu_1$ & sd($\nu_1$) & $\nu_2$ & sd($\nu_2$)\\
\hline
\hline
\multicolumn{7}{|c|}{Mixture 1 : $n=10^4$ $\lambda^* = 0.7$, $\nu_1^*=2$, $\sigma_1^*=0.5$(fixed), $\nu_2^*=1$, $\sigma_2^*=1$(fixed) }\\
\hline
Pearson's $\chi^2$ 3 moments & 0.700 & 0.010 & 2.006 & 0.217 & 1.005 & 0.024 \\
Pearson's $\chi^2$ 4 moments & 0.701 & 0.010 & 2.014 & 0.086 & 1.013 & 0.024 \\
Robin et al. \cite{Robin} & 0.654 & 0.101 & 1.591 & 0.085 & --- & --- \\
Song et al. EM-type \cite{Song} & 0.907 & 0.004 & 1.675 & 0.020 & --- & --- \\
Song et al. $\pi-$maximizing \cite{Song} & 0.782 & 0.006 & 1.443 & 0.012 & --- & --- \\
\hline
\hline
\multicolumn{7}{|c|}{Mixture 1 : $n=10^4$ $\lambda^* = 0.3$, $\nu_1^*=2$, $\sigma_1^*=0.5$(fixed), $\nu_2^*=1$, $\sigma_2^*=1$(fixed) }\\
\hline
Pearson's $\chi^2$ 3 moments & 0.304 & 0.016 & 2.191 & 0.887 & 0.998 & 0.013 \\
Pearson's $\chi^2$ 4 moments & 0.303 & 0.016 & 2.120 & 0.285 & 1.001 & 0.013 \\
Robin et al. \cite{Robin} & 0.604 & 0.029 & 1.256  & 0.037 & --- & --- \\
Song et al. EM-type \cite{Song} & 0.806 & 0.005 & 1.185 & 0.018 & --- & --- \\
Song et al. $\pi-$maximizing \cite{Song} & 0.624 & 0.007 & 1.312 & 0.013 & --- & --- \\
\hline
\end{tabular}
\caption{The mean value with the standard deviation of estimates in a 100-run experiment on a two-component Weibull mixture.}
\label{tab:3by3ResultsWeibullMoment}
\end{table}

\subsection{Data generated from a two-sided Weibull Gaussian mixture modeled by a semiparametric two-sided Weibull Gaussian mixture}\label{subsec:TwoSidGaussMom}
The (symmetric) two-sided Weibull distribution can be considered as a generalization of the Laplace distribution and can be defined through either its density or its distribution function as follows:
\[f(x|\nu,\sigma) = \frac{1}{2}\frac{\sigma}{\nu}\left(\frac{|x|}{\sigma}\right)^{\nu-1}e^{-\left(\frac{|x|}{\sigma}\right)^{\nu}}, \qquad \mathbb{F}(x|\nu,\sigma) = \left\{ \begin{array}{cc} 1-\frac{1}{2}e^{-\left(\frac{x}{\sigma}\right)^{\nu}} & x\geq 0 \\ 
e^{-\left(\frac{-x}{\sigma}\right)^{\nu}} & x< 0\end{array} \right.\]
We can also define a skewed form of the two-sided Weibull distribution by attributing different scale and shape parameters to the positive and the negative parts, and then normalizing in a suitable way so that $f(x)$ integrates to one; see \cite{Chen2sideWeibull}. The moments of the symmetric two-sided Weibull distribution we consider here are given by
\begin{eqnarray*}
\mathbb{E}[X^{2k}] & = & \sigma^{2k}\Gamma(1+2k/\nu) \\
\mathbb{E}[X^{2k+1}] & = & 0, \forall k\in\mathbb{N}.
\end{eqnarray*}
%\begin{figure}[h]
%\centering
%\includegraphics[scale=0.4]{TwoSidedWeibullGaussMixures.png}
%\caption{Mixtures of two-sided Weibull -- Gaussian with low and high proportion of the parametric part. See table (\ref{tab:3by3ResultsTwoSideWeibullGaussMom})}
%\label{fig:TwoSidedWeibullGaussMixure}
%\end{figure}

\noindent We simulate different samples from a two-component mixture with a parametric component $f_1$ a Gaussian $\mathcal{N}(\mu=0,\sigma=0.5)$ and a semiparametric component $f_0$ a (symmetric) two-sided Weibull distribution with parameters $\nu\in\{3,2.5,1.5\}$ and a scale $\sigma_0\in\{1.5,2\}$. We perform different experiments to estimate the proportion and the mean of the parametric part (the Gaussian) and the shape of the semiparametric component. The values of the scale of the two components are considered to be known during estimation. We consider the following two sets of constraints:
\begin{eqnarray*}
\mathcal{M}_{1:3} & = & \left\{f_0:\int_{\mathbb{R}}{f_0(x)dx}=1,\mathbb{E}_{f_0}[X] = 0, \mathbb{E}_{f_0}[X^2]=\sigma_0^2\Gamma(1+2/\nu), \mathbb{E}_{f_0}[X^3]=0, \nu>0\right\}; \\
\mathcal{M}_{2:4} & = & \left\{f_0: \int_{\mathbb{R}}{f_0(x)dx}=1,\mathbb{E}_{f_0}[X^2]=\sigma_0^2\Gamma(1+2/\nu), \mathbb{E}_{f_0}[X^3]=0,\right. \\
&  & \text{\hspace{8cm}} \left.  \mathbb{E}_{f_0}[X^4]=\sigma_0^4\Gamma(1+4/\nu), \nu>0\right\}.
\end{eqnarray*}
The first set of constraints is not really suitable for estimation especially when the number of observations is high enough. The reason is simple and is based on the original idea behind our procedure, see paragraph \ref{subsec:procIntrod}. The first and the third moment constraints are practically the same constraint. Indeed, the number of models of the form $\frac{1}{1-\lambda}f(.) - \frac{\lambda}{1-\lambda}f_1(x|\theta)$ verifying the constraints of $\mathcal{M}_{1:3}$ is infinite because the first and the third constraints give rise to the following equations:
\begin{eqnarray*}
\lambda\mu & = & 0 \\
\lambda\mu\left(\mu^2+3\sigma_1^2\right) & = & 0.
\end{eqnarray*}
The zero in the right hand side comes from the fact that the first and the third true moments of the whole mixture are zero. These are two equations in $\lambda$ and $\mu$ (since $\sigma_1$ is supposed to be known) with infinite number of solutions $(\mu,\lambda)\in\{0\}\times[0,1]$. This entails that theoretically, there is an infinite number of models of the form $\frac{1}{1-\lambda}f(.) - \frac{\lambda}{1-\lambda}f_1(x|\theta)$ in the intersection $\mathcal{N}\cap M_{1:3}$. Still, the empirical version of these equations is
\begin{eqnarray*}
\lambda\mu & = & \frac{1}{n}\sum_{i=1}^n{X_i}; \\
\lambda\mu\left(\mu^2+3\sigma_1^2\right) & = & \frac{1}{n}\sum_{i=1}^n{X_i^3}.
\end{eqnarray*}
As the number of observations is very small, the right hand side of both equations is biased enough from zero and it is highly possible that the number of solutions becomes not only finite but reduced to one. As the number of observations increases, the law of large numbers implies directly that the right hand side becomes arbitrarily close to zero and the set of solutions becomes infinite. This is exactly what happened in the simulation results in Table \ref{tab:3by3ResultsTwoSideWeibullGaussMom} below. The algorithm favored the value zero for the estimate of the proportion as the true proportion of the parametric component became close to zero, whereas the estimates of the mean took values very dispersed centered around zero but with a high standard deviation. The set of constraints $\mathcal{M}_{2:4}$ gave clear better results even for very low proportions. On the other hand, our method outperforms other semiparametric algorithms without prior information especially when the proportion of the parametric component is low. This shows once more the interest of incorporating a prior information in the estimation procedure.
\begin{table}[ht]
\centering
\resizebox{\textwidth}{!}{ 
\begin{tabular}{|c|c|c|c|c|c|c|}
\hline
Estimation method & $\lambda$ & sd$(\lambda)$ & $\mu$ & sd($\mu$) & $\nu$ & sd($\nu$)\\
\hline
\hline
\multicolumn{7}{|c|}{Mixture 1 :$n=100$ $\lambda^* = 0.7$, $\mu^*=0$, $\sigma_2^*=0.5$(fixed), $\nu^*=3$, $\sigma_1^*=1.5$(fixed) }\\
\hline
Pearson's $\chi^2$ under $\mathcal{M}_{1:3}$ & 0.713 & 0.064 & -0.0003 & 0.085 & 4.315  & 0.118 \\
Pearson's $\chi^2$ under $\mathcal{M}_{2:4}$ & 0.764 & 0.067 & -0.012 & 0.342 & 2.893  & 0.731 \\
Bordes et al. symmetry Triangular Kernel \cite{Bordes10} & 0.309 & 0.226 & 0.240 & 0.609 & $\mu_2=-$0.220 & sd$(\mu_2)=$0.398 \\
Bordes et al. symmetry Gaussian Kernel \cite{Bordes10} & 0.211 & 0.133 & 0.106 & 0.533 & $\mu_2=-$0.035 & sd$(\mu_2)=$0.203 \\
Robin et al. \cite{Robin} & 0.488 & 0.137 & -0.005 & 0.114 & --- & --- \\
EM-type Song et al. \cite{Song} & 0.762 & 0.040 & -0.005 & 0.092 & --- & --- \\
$\pi-$maximizing Song et al. \cite{Song} & 0.717 & 0.156 & -0.161 & 2.301 & --- & --- \\
Stochastic EM \cite{BordesStochEM} & 0.539 & 0.083 & -0.005 & 0.112 & --- & --- \\
\hline
\hline
\multicolumn{7}{|c|}{Mixture 2 :$n=100$ $\lambda^* = 0.3$, $\mu^*=0$, $\sigma_2^*=0.5$(fixed), $\nu^*=3$, $\sigma_1^*=1.5$(fixed) }\\
\hline
Pearson's $\chi^2$ under $\mathcal{M}_{1:3}$ & 0.333 & 0.079 & 0.001 & 0.316 & 4.243  & 0.442 \\
Pearson's $\chi^2$ under $\mathcal{M}_{2:4}$ & 0.407 & 0.077 & 0.012 & 0.575 & 2.925  & 0.454 \\
Bordes et al. symmetry Triangular Kernel \cite{Bordes10} & 0.272 & 0.119 & 0.773 & 0.947 & $\mu_2=-$0.430 & sd$(\mu_2)=$0.393 \\
Bordes et al. symmetry Gaussian Kernel \cite{Bordes10} & 0.206 & 0.104 & 0.855 & 0.911 & $\mu_2=-$0.308 & sd$(\mu_2)=$0.350 \\
Robin et al. \cite{Robin} & 0.203 & 0.078 & -0.109 & 0.947 & --- & --- \\
EM-type Song et al. \cite{Song} & 0.494 & 0.035 & -0.132 & 0.806 & --- & --- \\
$\pi-$maximizing Song et al. \cite{Song} & 0.384 & 0.129 & 0.014 & 1.321 & --- & --- \\
Stochastic EM \cite{BordesStochEM} & 0.263 & 0.040 & -0.062 & 0.646 & --- & --- \\
\hline
\hline
\multicolumn{7}{|c|}{Mixture 6 :$n=10^5$ $\lambda^* = 0.05$, $\mu^*=0$, $\sigma_2^*=0.5$(fixed), $\nu^*=1.5$, $\sigma_1^*=2$(fixed) }\\
\hline
Pearson's $\chi^2$ under $\mathcal{M}_{1:3}$ & 0.005 & 0.033 & -0.105 & 2.693 & 1.581  & 0.056 \\
Pearson's $\chi^2$ under $\mathcal{M}_{2:4}$ & 0.066 & 0.013 & -0.036 & 0.857 & 1.493 & 0.008 \\
EM-type Song et al.\cite{Song} & 0.304 & 0.014 & -0.030 & 0.910 & --- & --- \\
$\pi-$maximizing Song et al. \cite{Song} & 0.231 & 0.002 & 0.017 & 0.801 & --- & ---\\
\hline
\hline
\multicolumn{7}{|c|}{Mixture 7 :$n=10^7$ $\lambda^* = 0.05$, $\mu^*=0$, $\sigma_2^*=0.5$(fixed), $\nu^*=1.5$, $\sigma_1^*=2$(fixed) }\\
\hline
Pearson's $\chi^2$ under $\mathcal{M}_{1:3}$ & 0.006 & 0.010 & 0.024 & 0.197 & 1.500  & 0.019 \\
Pearson's $\chi^2$ under $\mathcal{M}_{2:4}$ & 0.051 & 0.001 & 0.002 & 0.259 & 1.500  & 0.001 \\
\hline
\hline
\multicolumn{7}{|c|}{Mixture 8 :$n=10^7$ $\lambda^* = 0.01$, $\mu^*=0$, $\sigma_2^*=0.5$(fixed), $\nu^*=1.5$, $\sigma_1^*=2$(fixed) }\\
\hline
Pearson's $\chi^2$ under $\mathcal{M}_{1:3}$ & 0.005 & 0.002 & -0.011 & 0.162 & 1.509  & 0.004 \\
Pearson's $\chi^2$ under $\mathcal{M}_{2:4}$ & 0.011 & 0.001 & -0.013 & 0.594 & 1.499  & 0.001 \\
\hline
\end{tabular}}
\caption{The mean value with the standard deviation of estimates in a 100-run experiment on a two-component two-sided Weibull--Gaussian mixture.}
\label{tab:3by3ResultsTwoSideWeibullGaussMom}
\end{table}

%%%%%%%%%%%%%%%%%%%%%%%%%%%%%%%%%%
\subsection{Data generated from a bivariate Gaussian mixture and modeled by a semiparametric bivariate Gaussian mixture}
We generate 1000 i.i.d. observations from a bivariate Gaussian mixture with proportion $\lambda=0.7$ for the parametric component. The parametric component is a bivariate Gaussian with mean $(0,-1)$ and covariance matrix $I_2$. The unknown component is a bivariate Gaussian with mean $(3,3)$ and covariance matrix:
\[\Sigma_2 = \left(\begin{array}{cc}{\sigma_2^*}^2 & \rho^* \\ \rho^* & {\sigma_2^*}^2\end{array}\right),\qquad {\sigma_2^*}^2 = 0.5,\quad \rho^* \in\{0, 0.25\}.\]
In a first experiment, we suppose that we know the whole parametric component, and that the unknown component belongs to the set $\mathcal{M}_1$
\begin{equation*}
\mathcal{M}_1 = \left\{\int_{\mathbb{R}^2}{f_0(x,y)dxdy}=1,\quad \int_{\mathbb{R}^2}{xf_0(x,y)dxdy}=\int_{\mathbb{R}^2}{yf_0(x,y)dydx}=\theta,\quad \theta\in\mathbb{R}\right\}.
\end{equation*}
We suppose that the only unknown parameters are the center of the unknown cluster $(\theta,\theta)$ and the proportion of the parametric component.\\
In a second experiment, we suppose that the center of the parametric component is unknown but given by $(\mu,\mu-1)$ for some unknown $\mu\in\mathbb{R}$. The set of constraints is now replaced with $\mathcal{M}_2$ given by
\begin{eqnarray*}
\mathcal{M}_2 & = &  \left\{\int_{\mathbb{R}^2}{f_0(x,y)dxdy}=1,\quad \int_{\mathbb{R}^2}{xf_0(x,y)dxdy}=\int_{\mathbb{R}^2}{yf_0(x,y)dydx}=\theta, \right.\\
 &  & \left.\qquad \qquad \qquad \int_{\mathbb{R}^2}{xyf_0(x,y)dxdy}=\theta^2+\rho^*,\theta\in\mathbb{R}\right\}.
\end{eqnarray*}
The covariance between the two coordinates $\rho^*$ in the unknown component is supposed to be known. We tested two values for $\rho^*= 0$ and $\rho^*=0.25$.\\
Although existing methods were only proposed for univariate cases, we see no problem in using them in multivariate cases without any changes. The only method which cannot be used directly is the method of \cite{Bordes10} because it is based on the symmetry of the density function, so it remained out of the competition.\\
For methods which use a kernel estimator, we used a kernel estimator for each coordinate of the random observations, i.e. $K_{w_x,w_y}(x,y) = K_{w_x}(x)K_{w_y}(y)$. The EM-type algorithm of \cite{Song} performs as good as our algorithm. The SEM algorithm of \cite{BordesStochEM} gives also good results. The algorithm of \cite{Robin} and the $\pi-$maximizing algorithm of \cite{Song} failed to give satisfactory results.
%\begin{figure}[h]
%\centering
%\includegraphics[scale=0.35]{BivariateGaussMix.png}
%\caption{The two bivariate Gaussian mixtures.}
%\label{fig:BivariateGaussMix}
%\end{figure}

\begin{table}[h]
\centering
\begin{tabular}{|c|c|c|c|c|c|c|}
\hline
Estimation method & $\lambda$ & sd$(\lambda)$ & $\mu$ & sd($\mu$) & $\theta$ & sd($\theta$)\\
\hline
\hline
\multicolumn{7}{|c|}{\bf{Mixture 1} : $\rho^*=0$ and $\mu_1 = (\mu,1-\mu)$ is unknown }\\
\hline
Pearson's $\chi^2$ under $\mathcal{M}_1$ & 0.680 & 0.027 & --- & --- & 2.854 & 0.233 \\
Pearson's $\chi^2$ under $\mathcal{M}_2$ & 0.694 & 0.019 & 0.016 & 0.035 & 3.034 & 0.045 \\
Stochastic EM \cite{BordesStochEM}  & 0.724 & 0.015 & 0.090 & 0.043 & $\mu_{1,2}=$ -0.880 & sd$(\mu_{1,2})=$0.053 \\
Robin et al. \cite{Robin} & 0.954 & 0.064 & 0.779 & 0.212 & $\mu_{1,2}=$ -0.221 & sd$(\mu_{1,2})=$0.218 \\
Song EM-type \cite{Song}  & 0.697 & 0.014 & 0.003 & 0.038 & $\mu_{1,2}=$ -0.996 & sd$(\mu_{1,2})=$0.039 \\
Song $\pi-$maximizing \cite{Song} & 0.114 & 0.297 & 0.538 & 1.810 & $\mu_{1,2}=$ -0.463 & sd$(\mu_{1,2})=$1.810 \\
\hline
\hline
\multicolumn{7}{|c|}{\bf{Mixture 2} : $\rho^*=0.25$ and $\mu_1 = (\mu,1-\mu)$ is unknown }\\
\hline
Pearson's $\chi^2$ under $\mathcal{M}_2$ & 0.704 & 0.026 & 0.033 & 0.060 & 3.071 & 0.101 \\
Stochastic EM \cite{BordesStochEM} & 0.730 & 0.016 & 0.083 & 0.052 & $\mu_{1,2}=$-0.878 & sd$(\mu_{1,2})=$0.055 \\
Robin et al. \cite{Robin} & 0.890 & 0.025 & 0.566 & 0.117 & $\mu_{1,2}=$ -0.434 & sd$(\mu_{1,2})=$0.117 \\
Song EM-type \cite{Song}  & 0.704 & 0.015 & 0.016 & 0.047 & $\mu_{1,2}=$ -0.973 & sd$(\mu_{1,2})=$0.040 \\
Song $\pi-$maximizing \cite{Song} & 0.095 & 0.268 & 0.564 & 1.606 & $\mu_{1,2}=$ -0.436 & sd$(\mu_{1,2})=$1.606 \\
\hline

\end{tabular}
\caption{The mean value with the standard deviation of estimates in a 100-run experiment on a two-component bivariate normal mixture.}
\label{tab:3by3ResultsBivaraiteGaussMom}
\end{table}

% ==========================================================
% -------------------------------------------------
%%%%%%%%%%%%%%%%%%%%%%%%%%%%%%%%%%%%%%%%%%%%%%%%%%%%%%%%%%%%%%%%%%%%%%%
% -------------------------------------------------
% ==========================================================
\section{Conclusions}
In this paper, we proposed a structure for semiparametric two-component mixture models where one component is parametric with unknown parameter, and a component defined by linear constraints on its distribution function. These constraints may be moments constraints for example. We proposed also an algorithm which estimates the parameters of this model and showed how we can implement it efficiently even in multivariate contexts. The algorithm has a linear complexity when we use the Pearson's $\chi^2$ divergence and the constraints are polynomials (thus moments constraints). We provided sufficient conditions in order to prove the consistency and the asymptotic normality of the resulting estimators.\\
Simulations show the gain we have by adding moments constraints in comparison to existing methods which do not consider any prior information. The method give clear good results even if the proportion of the parametric component is very low (equal to 0.01). In signal-noise applications, this can be interpreted otherwise. As long as we are able to estimate with relatively high precision the proportion of the signal (parametric component), we are proving the existence of the signal in a very heavy noise ($99\%$ of the data) even if the position of the signal is not accurately estimated. We showed in a simple example that our model can be applied in multivariate contexts. The new model shows encouraging properties and results, and should be tested on real data in a future work.

\clearpage
% ==========================================================
% -------------------------------------------------
%%%%%%%%%%%%%%%%%%%%%%%%%%%%%%%%%%%%%%%%%%%%%%%%%%%%%%%%%%%%%%%%%%%%%%%
% -------------------------------------------------
% ==========================================================
\section{Appendix: Proofs}
\subsection{Proof of Proposition \ref{prop:identifiabilityMixture}}\label{AppendSemiPara:PropIdenitifiability}
\begin{proof}
Based on equation (\ref{eqn:IdenitifiabilityDefEq}), we may write the corresponding constraints equations, which are a fortiori equal:
\[\lambda \int{g(x)dP_1(x|\theta)} + (1-\lambda)m(\alpha) = \tilde{\lambda} \int{g(x)dP_1(x|\tilde{\theta})} + (1-\tilde{\lambda})m(\tilde{\alpha}).\]
Define the following function:
\[G:\mathbb{R}^d\rightarrow\mathbb{R}^{\ell}: (\lambda,\theta,\alpha)\mapsto \lambda \int{g(x)dP_1(x|\theta)} + (1-\lambda)m(\alpha).\]
The solution to the previous system of equations is now equivalent to the fact that function $G$ is one-to-one. This means that for a fixed $m^*$, we need that the nonlinear system of equations:
\begin{equation}
\frac{1}{1-\lambda} m^* - \frac{\lambda}{1-\lambda}m_1(\theta) = m_0(\alpha)
\label{eqn:IdentSys}
\end{equation}
has a unique solution $(\lambda,\theta,\alpha)$. The value of $m^*$ is given by $\int{g(x)dP_T}$ where $P_T$ is the mixture we are considering. To conclude, suppose that the system (\ref{eqn:IdentSys}) has a unique solution $(\lambda^*,\theta^*,\alpha^*)$ for each given $m^*$, then function $G$ is one-to-one and the constraints equations imply that $\lambda = \tilde{\lambda},\theta = \tilde{\theta}$ and $\alpha=\tilde{\alpha}$. Finally, using (\ref{eqn:IdenitifiabilityDefEq}), we may deduce that $P_0 = \tilde{P}_0$. Thus, the semiparametric mixture model is identifiable as soon as the nonlinear system of equations (\ref{eqn:IdentSys}) has a unique solution $(\lambda^*,\theta^*,\alpha^*)$.
\end{proof}

%%%%%%%%%%%%%%%%%%%%%%%%%%%%%%%%%%%%%%%%%%%%%%%%%%%%%%%
%%%%%%%%%%%%%%%%%%%%%%%%%%%%%%%%%%%%%%%%%%%%%%%%%%%%%%%

\subsection{Proof of Proposition \ref{prop:identifiability}}\label{AppendSemiPara:Prop1}
\begin{proof}
Let $P_0$ be some signed measure which belongs to the intersection $\mathcal{N} \cap \mathcal{M}$. Since $P_0$ belongs to $\mathcal{N}$, there exists a couple $(\lambda,\theta)$ such that:
\begin{equation}
P_0 = \frac{1}{1-\lambda} P_T - \frac{\lambda}{1-\lambda} P_1(.|\theta).
\label{eqn:SetNelement}
\end{equation}
This couple is unique by virtue of assumptions 3 and 4. Indeed, let $(\lambda,\theta)$	and $(\tilde{\lambda},\tilde{\theta})$ be two couples such that:
\begin{equation}
\frac{1}{1-\lambda} P_T - \frac{\lambda}{1-\lambda} P_1(.|\theta) = \frac{1}{1-\tilde{\lambda}} P_T - \frac{\tilde{\lambda}}{1-\tilde{\lambda}} P_1(.|\tilde{\theta})\quad dP_T-a.e.
\label{eqn:identifEquality}
\end{equation}
This entails that:
\[\frac{1}{1-\lambda} - \frac{\lambda}{1-\lambda} \frac{dP_1(x|\theta)}{dP_T(x)} = \frac{1}{1-\tilde{\lambda}} - \frac{\tilde{\lambda}}{1-\tilde{\lambda}} \frac{dP_1(x|\tilde{\theta})}{dP_T(x)}.\]
Taking the limit as $\|x\|$ tends to $\infty$ results in:
\[\frac{1-c\lambda}{1-\lambda}  = \frac{1-\tilde{c}\lambda}{1-\tilde{\lambda}}.\]
Note that function $z\mapsto (1-cz)/(1-z)$ is strictly monotone as long as $c\neq 1$. Hence, it is a one-to-one map. Thus $\lambda=\tilde{\lambda}$. Inserting this result in equation (\ref{eqn:identifEquality}) entails that:
\[P_1(.|\theta) = P_1(.|\tilde{\theta})\qquad dP_T-a.e.\]
Using the identifiability of $P_1$ (assumption 4), we get $\theta=\tilde{\theta}$ which proves the existence of a unique couple $(\lambda,\theta)$ in (\ref{eqn:SetNelement}).\\
On the other hand, since $P_0$ belongs to $\mathcal{M}$, there exists a unique $\alpha$ such that $P_0\in\mathcal{M}_{\alpha}$. Uniqueness comes from the fact that function $\alpha\mapsto m(\alpha)$ is one-to-one (assumption 2). Thus, $P_0$ verifies the constraints
\[\int{dP_0(x)} = 1,\qquad \int{g_i(x)dP_0(x)} = m_i(\alpha),\quad \forall i=1,\cdots,\ell.\]
Combining this with (\ref{eqn:SetNelement}), we get:
\begin{equation}
\int{\left(\frac{1}{1-\lambda} dP_T - \frac{\lambda}{1-\lambda} dP_1(x|\theta)\right)} = 1,\; \int{g_i(x)\left(\frac{1}{1-\lambda} dP_T - \frac{\lambda}{1-\lambda} dP_1(x|\theta)\right)} = m_i(\alpha),
\label{eqn:NlnSysMalpha}
\end{equation}
for all $i=1,\cdots,\ell$. This is a non linear system of equations with $\ell+1$ equations. The first one is verified for any couple $(\lambda,\theta)$ since both $P(.|\phi^*)$ and $P_1$ are probability measures. This reduces the system to $\ell$ nonlinear equations.\\
Now, let $P_0$ and $\tilde{P}_0$ be two elements in $\mathcal{N}\cap\mathcal{M}$, then there exist two couples $(\lambda,\theta)$ and $(\tilde{\lambda},\tilde{\theta})$ with $\lambda\neq\tilde{\lambda}$ or $\theta\neq\tilde{\theta}$. Since $P_0\in\mathcal{M}$, there exists $\alpha$ such that $P_0\in\mathcal{M}_{\alpha}$. Similarly, there exists $\tilde{\alpha}$ possibly different from $\alpha$. Now, $(\lambda,\theta,\alpha)$ and $(\tilde{\lambda},\tilde{\theta},\tilde{\alpha})$ are two solutions to the system of equations (\ref{eqn:NlnSysMalpha}) which contradicts with assumption 1 of the present proposition.\\
We may now conclude that, if a signed measure $P_0$ belongs to the intersection $\mathcal{N} \cap \mathcal{M}$, then it has the representation (\ref{eqn:SetNelement}) for a unique couple $(\lambda,\theta)$ and there exists a unique $\alpha$ such that the triplet $(\lambda,\theta,\alpha)$ is a solution to the non linear system (\ref{eqn:NlnSysMalpha}). Conversely, if there exists a triplet $(\lambda,\theta,\alpha)$ which solves the non linear system (\ref{eqn:NlnSysMalpha}), then the signed measure $P_0$ defined by $P_0 = \frac{1}{1-\lambda} P(.|\phi^*) - \frac{\lambda}{1-\lambda} P_1(.|\theta)$ belongs to the intersection $\mathcal{N} \cap \mathcal{M}$. This is because on the one hand, it clearly belongs to $\mathcal{N}$ by its definition and on the other hand, it belongs to $\mathcal{M}_{\alpha}$ since it verifies the constraints and thus belongs to $\mathcal{M}$.\\
It is now reasonable to conclude that under assumptions 2-4, the intersection $\mathcal{N} \cap \mathcal{M}$ includes a \emph{unique} signed measure $P_0$ if and only if the set of $\ell$ non linear equations (\ref{eqn:NlnSysMalpha}) has a unique solution $(\lambda,\theta,\alpha)$.
\end{proof}

% -------------------------------------------------
%%%%%%%%%%%%%%%%%%%%%%%%%%%%%%%%%%%%%%%%%%%%%%%%%%%%%%%%%%%%%%%%%%%%%%%
% -------------------------------------------------
\subsection{Proof of Lemma \ref{lem:SupXiPhiDiff}}\label{AppendSemiPara:Lem1}
\begin{proof}
The proof is based partially on the proof of Proposition 3.7 part (ii) in \cite{KeziouThesis}.\\
We proceed by contradiction. Let $\varepsilon>0$ be such that $\sup_{\phi}\|\xi_n(\phi) - \xi(\phi)\|>\varepsilon$. Then, there exists a sequence $a_k\in\Phi$ such that $\|\xi_n(a_k) - \xi(a_k)\|>\varepsilon$. By assumption A3, there exists $\eta>0$ such that:
\[H(a_k,\xi(a_k)) - H(a_k,\xi_n(a_k))>\eta.\]
Thus,
\begin{equation}
\mathbb{P}\left(\sup_{\phi}\|\xi_n(\phi) - \xi(\phi)\|>\varepsilon\right) \leq \mathbb{P}\left(H(a_k,\xi(a_k)) - H(a_k,\xi_n(a_k))>\eta\right).
\label{eqn:ProofPart1InteriorConsis}
\end{equation}
Let's prove that the right hand side tends to zero as $n$ goes to infinity which is sufficient to accomplish our claim.\\
By definition of $\xi_n(a_k)$ and assumption A2, we can write:
\begin{eqnarray*}
H_n(a_k,\xi_n(a_k)) & \geq & H_n(a_k,\xi(a_k)) \\
	& \geq & H(a_k,\xi(a_k)) - o_P(1)
\end{eqnarray*}
where $o_P(1)$ does not depend upon $a_k$ by virtue of A2. Now we have:
\begin{eqnarray*}
H(a_k,,\xi(a_k)) - H(a_k,\xi_n(a_k)) & \leq & H_n(a_k,\xi_n(a_k)) - H(a_k,,\xi_n(a_k)) + o_P(1) \\
 & \leq & \sup_{\xi,\phi} \left|H_n(\phi,\xi) - H(\phi,\xi)\right| + o_P(1).
\end{eqnarray*}
Last but not least, assumption A2 permits to conclude that the right hand side tends to zero in probability. Since the left hand side is already nonnegative by definition of $\xi(a_k)$, then by the previous result we conclude that $H(a_k,,\xi(a_k)) - H(a_k,\xi_n(a_k))$ tends to zero in probability. Employing this final result in inequality (\ref{eqn:ProofPart1InteriorConsis}), we get that $\sup_{\phi}\|\xi_n(\phi) - \xi(\phi)\|$ tends to zero in probability.
\end{proof}

% -------------------------------------------------
%%%%%%%%%%%%%%%%%%%%%%%%%%%%%%%%%%%%%%%%%%%%%%%%%%%%%%%%%%%%%%%%%%%%%%%
% -------------------------------------------------
\subsection{Proof of Theorem \ref{theo:MainTheorem}}\label{AppendSemiPara:Theo1}
\begin{proof}
We proceed by contradiction in a similar way to the proof of Lemma \ref{lem:SupXiPhiDiff}. Let $\kappa>0$ be such that $\|\phi^*-\hat{\phi}\|>\kappa$, then by assumption A4, there exists $\eta>0$ such that :
\[H(\hat{\phi},\xi(\hat{\phi})) - H(\phi^*,\xi(\phi^*)) > \eta.\]
This can be rewritten as:
\begin{equation}
\mathbb{P}\left(\|\phi^*-\hat{\phi}\|>\kappa\right) \leq \mathbb{P}\left(H(\hat{\phi},\xi(\hat{\phi})) - H(\phi^*,\xi(\phi^*)) > \eta\right).
\label{eqn:ProofPart2ExtConsis}
\end{equation}
We now demonstrate that the right hand side tends to zero as $n$ goes to infinity. Let $\varepsilon>0$ be such that for $n$ sufficiently large, we have $\sup_{\xi,\phi} \left|H(\phi,\xi)-H_n(\phi,\xi)\right|<\varepsilon$. This is possible by virtue of assumption A2. The definition of $\hat{\phi}$ together with assumption A2 will now imply:
\begin{eqnarray}
H_n(\hat{\phi},\xi_n(\hat{\phi})) & \leq & H_n(\phi^*,\xi_n(\phi^*)) \nonumber\\
  & \leq & H(\phi^*,\xi_n(\phi^*)) + \sup_{\xi,\phi} \left|H(\phi,\xi)-H_n(\phi,\xi)\right| \nonumber\\
	& \leq & H(\phi^*,\xi_n(\phi^*)) + \varepsilon.
	\label{eqn:ProofPart2ExtConsisIneq}
\end{eqnarray}
We use now the continuity assumption A5 of function $\xi\mapsto H(\phi^*,\xi)$ at $\xi(\phi^*)$. For the $\varepsilon$ chosen earlier, there exists $\delta(\phi^*,\varepsilon)$ such that if $\|\xi(\phi^*)-\xi_n(\phi^*)\|<\delta(\phi^*,\varepsilon)$, then:
\[|H(\phi^*,\xi_n(\phi^*)) - H(\phi^*,\xi(\phi^*))|<\varepsilon.\]
This is possible for sufficiently large $n$ since $\sup_{\phi}\|\xi(\phi^*)-\xi_n(\phi^*)\|$ tends to zero in probability by Lemma \ref{lem:SupXiPhiDiff}. Inserting this result in (\ref{eqn:ProofPart2ExtConsisIneq}) gives:
\[H_n(\hat{\phi},\xi_n(\hat{\phi})) \leq H(\phi^*,\xi(\phi^*)) + 2\varepsilon.\]
We now have:
\begin{eqnarray*}
H(\hat{\phi},\xi(\hat{\phi})) - H(\phi^*,\xi(\phi^*)) & \leq & H(\hat{\phi},\xi(\hat{\phi})) - H_n(\hat{\phi},\xi_n(\hat{\phi})) + 2\varepsilon \\
 & \leq & H(\hat{\phi},\xi(\hat{\phi})) - H(\hat{\phi},\xi_n(\hat{\phi})) + H(\hat{\phi},\xi_n(\hat{\phi})) - H_n(\hat{\phi},\xi_n(\hat{\phi})) + 2\varepsilon.
\end{eqnarray*}
Continuity assumption of $H$ implies that for $\varepsilon>0$, there exists $\delta(\hat{\phi},\varepsilon)>0$ such that if $\|\xi(\hat{\phi}) - \xi_n(\hat{\phi})\|<\delta(\hat{\phi},\varepsilon)$, then:
\[\left|H(\hat{\phi},\xi(\hat{\phi})) - H(\hat{\phi},\xi_n(\hat{\phi}))\right| \leq \varepsilon.\]
This is again possible for sufficiently large $n$ since $\sup_{\phi}\|\xi(\phi^*)-\xi_n(\phi^*)\|$ tends to zero in probability by Lemma \ref{lem:SupXiPhiDiff}. This entails that:
\begin{eqnarray*}
H(\hat{\phi},\xi(\hat{\phi})) - H(\phi^*,\xi(\phi^*)) & \leq & H(\hat{\phi},\xi_n(\hat{\phi})) - H_n(\hat{\phi},\xi_n(\hat{\phi})) + 3\varepsilon \\
  & \leq & \sup_{\xi,\phi} |H(\phi,\xi) - H_n(\phi,\xi)| + 3\varepsilon \\
	& \leq & 4\varepsilon
\end{eqnarray*}
We conclude that the right hand side in (\ref{eqn:ProofPart2ExtConsis}) goes to zero and the proof is completed.
\end{proof}

% -------------------------------------------------
%%%%%%%%%%%%%%%%%%%%%%%%%%%%%%%%%%%%%%%%%%%%%%%%%%%%%%%%%%%%%%%%%%%%%%%
% -------------------------------------------------
\subsection{Proof of Theorem \ref{theo:MainTheoremMomConstr}}\label{AppendSemiPara:Theo2}
\begin{proof}
We will use Theorem \ref{theo:MainTheorem}. We need to verify assumptions A2 and A3. Since the class of functions $\{(\phi,\xi)\mapsto h(\phi,\xi,.)\}$ is a Glivenko-Cantelli class of functions, then assumption A2 is fulfilled by the Glivenko-Cantelli theorem. Finally, assumption A3 can be checked by strict concavity of function $\xi\mapsto H(\phi,\xi)$. Indeed, for any $\eta\in(0,1)$ and any $\xi_1,\xi_2$, we have by strict convexity of $\psi$ :
\[\psi\left(\eta\xi_1^tg(x)+(1-\eta)\xi_2^tg(x)\right)<\eta\psi\left(\xi_1^tg(x)\right)+(1-\eta)\psi\left(\xi_2^tg(x)\right).\]
If the measure $dP/(1-\lambda) - \lambda dP_1(.|\theta)/(1-\lambda)$ is positive\footnote{This measure can never be zero since it integrates to one, thus we do not need to suppose that it is nonnegative.}, we may write:
\begin{multline*}
\int{\psi\left(\eta\xi_1^tg(x)+(1-\eta)\xi_2^tg(x)\right)\left(\frac{1}{1-\lambda}dP(x)-\frac{\lambda}{1-\lambda}dP_1(x|\theta)\right)}< \\ \eta\int{\psi\left(\xi_1^tg(x)\right)\left(\frac{1}{1-\lambda}dP(x)-\frac{\lambda}{1-\lambda}dP_1(x|\theta)\right)} + (1-\eta)\int{\psi\left(\xi_2^tg(x)\right)\left(\frac{1}{1-\lambda}dP(x)-\frac{\lambda}{1-\lambda}dP_1(x|\theta)\right)},
\end{multline*}
which entails that
\[H(\phi,\eta\xi_1+(1-\eta)\xi_2)> \eta H(\phi,\xi_1)+(1-\eta)H(\phi,\xi_2),\]
and function $\xi\mapsto H(\phi,\xi)$ becomes strictly concave. However, the measure $dP/(1-\lambda) - \lambda dP_1(.|\theta)/(1-\lambda)$ is in general a signed measure and the previous implication does not hold. This is not dramatic because function $\xi\mapsto H(\phi,\xi)$ has only two choices; it is either strictly convex or strictly concave. In case function $\xi\mapsto H(\phi,\xi)$ is strictly convex, then its supremum is infinity and the corresponding vector $\phi$ does not count in the calculus of the infimum after all. This means that the only vectors $\phi\in\Phi$ which interest us are those for which function $\xi\mapsto H(\phi,\xi)$ is strictly concave. In other words, the infimum in (\ref{eqn:MomentEstimProc}) can be calculated over the set:
\[\Phi^+ = \Phi\cap \left\{\phi: \quad \xi\mapsto H(\phi,\xi) \text{ is strictly concave}\right\}\]
instead of over $\Phi$. All assumptions of Theorem \ref{theo:MainTheorem} are now fulfilled and $\hat{\phi}$ converges in probability to $\phi^*$.
\end{proof} 

% -------------------------------------------------
%%%%%%%%%%%%%%%%%%%%%%%%%%%%%%%%%%%%%%%%%%%%%%%%%%%%%%%%%%%%%%%%%%%%%%%
% -------------------------------------------------
\subsection{Proof of Proposition \ref{prop:ContinDiffxiMom}}\label{AppendSemiPara:Prop2}
\begin{proof}
We already have:
\[\frac{1}{1-\lambda^*}P_T - \frac{\lambda^*}{1-\lambda^*}P_1(.|\theta^*)=P_0^*,\]
and since $P_0^*$ is supposed to be a probability measure, the matrix $J_{H(\phi^*,.)}$ is definite negative. Thus $\phi^*\in\Phi^+$. Since the set of negative definite matrices is an open set (see for example page 36 in \cite{OptimKenneth}), there exists a ball $\mathcal{U}$ of negative definite matrices centered at $J_{H(\phi^*,.)}$. Continuity of $\phi\mapsto J_{H(\phi,.)}$ permits\footnote{To see this, consider Sylvester's rule which is based on a test using the determinant of the sub-matrices of $J_{H}$. Each determinant needs to be negative. The continuity of the determinant function together with the continuity of $\phi\mapsto J_{H(\phi,.)}$ will imply that we may move around $J_(H(\phi^*,.))$ in a small neighborhood in a way that the determinants of the sub-matrices stay negative.} to find a ball $B(\phi^*,\tilde{r})$ such that the subset $\{J_{H(\phi,.)}: \phi\in B(\phi^*,\tilde{r})\}$ is inside $\mathcal{U}$. Now the neighborhood we are looking at is the ball $B(\phi^*,\tilde{r})$.\\
For the second part of the proposition, the existence and finiteness of $\xi(\phi)$ for $\phi\in\mathcal{V}=B(\phi^*,\tilde{r})$ is immediate since function $\xi\mapsto H(\phi,\xi)$ is strictly concave. Besides the the differentiability of the function $\phi\mapsto\xi(\phi)$ is a direct result of the implicit function theorem applied on the equation $\xi\mapsto \nabla H(\phi,.)$. Notice that the Hessian matrix of $H(\phi,.)$ is invertible since it is symmetric definite negative.
\end{proof}

% -------------------------------------------------
%%%%%%%%%%%%%%%%%%%%%%%%%%%%%%%%%%%%%%%%%%%%%%%%%%%%%%%%%%%%%%%%%%%%%%%
% -------------------------------------------------
\subsection{Proof of Theorem \ref{theo:AsymptotNormalMomConstr}}\label{AppendSemiPara:Theo3}
\begin{proof}
We follow the steps of Theorem 3.2 in \cite{NeweySmith}. The idea behind the proof is a mean value expansion with Lagrange remainder of the estimating equations. \\
We need at first to verify if $\hat{\phi}$ belongs to the interior of $\Phi^+$ in order to be able to differentiate $\phi\mapsto H_n(\phi,\xi)$. This can be done similarly to Proposition \ref{prop:ContinDiffxiMom}. We also can prove (by replacing $H$ by $H_n$ and $\xi(\phi)$ by $\xi_n(\phi)$) that $\phi\mapsto\xi_n(\phi)$ is continuously differentiable in a neighborhood of $\phi^*$.\\
We may now proceed to the mean value expansion. By the very definition of $\xi_n(\phi)$, we have:
\[\frac{\partial H_n}{\partial \xi}(\phi,\xi_n(\phi)) = 0\qquad \forall \phi\in\text{int}(\Phi^+),\]
which also holds for $\phi=\hat{\phi}$, i.e. 
\[\frac{\partial H_n}{\partial \xi}(\hat{\phi},\xi_n(\hat{\phi})) = 0.\]
On the other hand, the definition of $\hat{\phi}$ implies that:
\[\left.\frac{\partial}{\partial \phi}H_n(\phi,\xi_n(\phi))\right|_{\phi=\hat{\phi}} = 0.\]
Since function $\phi\mapsto\xi_n(\phi)$ is continuously differentiable. A simple chain rule implies
\begin{eqnarray*}
\left.\frac{\partial}{\partial \phi}\left(H_n(\phi,\xi_n(\phi))\right)\right|_{\phi=\hat{\phi}} & = &  \frac{\partial}{\partial \phi}H_n(\hat{\phi},\xi_n(\hat{\phi})) + \frac{\partial}{\partial \xi} H_n(\hat{\phi},\xi_n(\hat{\phi})) \frac{\partial \xi_n}{\partial \phi}(\hat{\phi}) \\
 & = & \frac{\partial}{\partial \phi}H_n(\hat{\phi},\xi_n(\hat{\phi})).
\end{eqnarray*}
The second line comes from the definition of $\xi_n(\phi)$ as the argument of the supremum of function $\xi\mapsto H_n(\phi,\xi)$. Now, the estimating equations are given simply by
\begin{eqnarray*}
\frac{\partial H_n}{\partial \xi}(\hat{\phi},\xi_n(\hat{\phi})) & = & 0; \\
\frac{\partial H_n}{\partial \phi}(\hat{\phi},\xi_n(\hat{\phi})) & = & 0.
\end{eqnarray*}
We need to calculate these partial derivatives. We start by the derivative with respect to $\xi$:
\begin{equation}
\frac{\partial H_n}{\partial \xi}(\phi,\xi) = m(\alpha) - \frac{1}{1-\lambda}\frac{1}{n}\sum_{i=1}^n{\psi'\left(\xi^tg(x_i)\right)g(x_i)} + \frac{\lambda}{1-\lambda}\int{\psi'\left(\xi^tg(x)\right)g(x)p_1(x|\theta)dx}
\label{eqn:DerivWRTxiMom}
\end{equation}
We calculate the partial derivatives with respect to $\alpha,\lambda$ and $\theta$:
\begin{eqnarray}
\frac{\partial H_n}{\partial \alpha} & = &  \xi^t \nabla m(\alpha) \label{eqn:DerivWRTalphaMom}\\
\frac{\partial H_n}{\partial \lambda} & = &  -\frac{1}{(1-\lambda)^2} \frac{1}{n}\sum_{i=1}^n{\psi\left(\xi^tg(x_i)\right)} + \frac{1}{(1-\lambda)^2}\int{\psi\left(\xi^tg(x)\right)p_1(x|\theta)dx} \label{eqn:DerivWRTlambdaMom}\\
\frac{\partial H_n}{\partial \theta} & = &  \frac{\lambda}{1-\lambda}\int{\psi\left(\xi^tg(x)\right)\nabla_{\theta} p_1(x|\theta)dx} \label{eqn:DerivWRTthetaMom}
\end{eqnarray}
Notice that by Lemma \ref{lem:SupXiPhiDiff}, the continuity of $\phi\mapsto\xi(\phi)$ and the consistency of $\hat{\phi}$ towards $\phi^*$, we have $\xi_n(\hat{\phi})\rightarrow \xi(\phi^*)=0$ in probability. A mean value expansion of the estimating equation between $(\hat{\phi},\xi_n(\hat{\phi}))$ and $(\phi^*,0)$ implies that there exists $(\bar{\phi},\bar{\xi})$ on the line between these two points such that:
\begin{equation}
\left(\begin{array}{c} \frac{\partial H_n}{\partial \phi}(\hat{\phi},\xi_n(\hat{\phi})) \\ \frac{\partial H_n}{\partial \xi}(\hat{\phi},\xi_n(\hat{\phi})) \end{array}\right) = \left(\begin{array}{c}  \frac{\partial H_n}{\partial \phi}(\phi^*,0) \\\frac{\partial H_n}{\partial \xi}(\phi^*,0) \end{array}\right)  + J_{H_n}(\bar{\phi},\bar{\xi}) \left(\begin{array}{c} \hat{\phi}-\phi^* \\ \xi_n(\hat{\phi})\end{array}\right),
\label{eqn:StochExpansion}
\end{equation}
where $J_{H_n}(\bar{\phi},\bar{\xi})$ is the matrix of second derivatives of $H_n$ calculated at the mid point $(\bar{\phi},\bar{\xi})$. The left hand side is zero, so we need to calculate the first vector in the right hand side. We have by simple substitution in formula (\ref{eqn:DerivWRTxiMom}):
\[\frac{\partial H_n}{\partial \xi}(\phi^*,0) = m(\alpha^*) - \frac{1}{1-\lambda^*}\frac{1}{n}\sum_{i=1}^n{g(x_i)} + \frac{\lambda^*}{1-\lambda^*}\int{g(x)p_1(x|\theta^*)dx}.\]
Using the assumption that the model (\ref{eqn:TrueP0model}) verify the set of constraints defining $\mathcal{M}_{\alpha}$ together with the CLT, we write:
\begin{equation}
\sqrt{n}\frac{\partial H_n}{\partial \xi}(\phi^*,0) \xrightarrow[\mathcal{L}]{} \mathcal{N}\left(0,\frac{1}{(1-\lambda^*)^2}\text{Var}_{P_T}(g(X))\right).
\label{eqn:LimitLawPartialDerivHn}
\end{equation}
Using formulas (\ref{eqn:DerivWRTalphaMom}), (\ref{eqn:DerivWRTlambdaMom}) and (\ref{eqn:DerivWRTthetaMom}), we may write:
\begin{eqnarray*}
\frac{\partial H_n}{\partial \alpha}(\phi^*,0) & = & 0; \\
\frac{\partial H_n}{\partial \lambda}(\phi^*,0) & = & -\frac{1}{(1-\lambda^*)^2} + \frac{1}{(1-\lambda^*)^2} = 0;\\
\frac{\partial H_n}{\partial \theta}(\phi^*,0) & = & \frac{\lambda^*}{1-\lambda^*}\int{\nabla_{\theta} p_1(x|\theta^*)dx} = \frac{\lambda^*}{1-\lambda^*}\nabla_{\theta} \int{p_1(x|\theta^*)dx} = 0.
\end{eqnarray*}
The final line holds since by Lebesgue's differentiability theorem using assumption 5 for $\xi=0$, we can change between the sign of integration and derivation. Combine this with the fact that $p_1(x|\theta^*)$ is a probability density function which integrates to 1, gives the result in the last line.\\
We need now to write an explicit form for the matrix $J_{H_n}(\bar{\phi},\bar{\xi})$ and study its limit in probability. It contains the second order partial derivatives of function $H_n$ with respect to its parameters. We start by the double derivatives. Using formulas (\ref{eqn:DerivWRTxiMom}), (\ref{eqn:DerivWRTalphaMom}), (\ref{eqn:DerivWRTlambdaMom}) and (\ref{eqn:DerivWRTthetaMom}), we write:
\begin{eqnarray*}
\frac{\partial^2H_n}{\partial \xi^2} & = & - \frac{1}{1-\lambda}\frac{1}{n}\sum_{i=1}^n{\psi''\left(\xi^tg(x_i)\right)g(x_i)g(x_i)^t} + \frac{\lambda}{1-\lambda}\int{\psi''\left(\xi^tg(x)\right)g(x)g(x)^tp_1(x|\theta)dx};\\
\frac{\partial^2H_n}{\partial \alpha^2} & = & \xi^t J_m(\alpha);\\
\frac{\partial^2H_n}{\partial \lambda^2} & = & -\frac{2}{(1-\lambda)^3} \frac{1}{n}\sum_{i=1}^n{\psi\left(\xi^tg(x_i)\right)} + \frac{2}{(1-\lambda)^3}\int{\psi\left(\xi^tg(x)\right)p_1(x|\theta)dx};\\
\frac{\partial^2H_n}{\partial \theta^2} & = & \frac{\lambda}{1-\lambda}\int{\psi\left(\xi^tg(x)\right)J_{p_1(x|\theta)}dx}; \\
\frac{\partial^2H_n}{\partial \xi\partial\alpha} & = & \nabla m(\alpha); \\
\frac{\partial^2H_n}{\partial \xi\partial\lambda} & = &  - \frac{1}{(1-\lambda)^2}\frac{1}{n}\sum_{i=1}^n{\psi'\left(\xi^tg(x_i)\right)g(x_i)} + \frac{1}{(1-\lambda)^2}\int{\psi'\left(\xi^tg(x)\right)g(x)p_1(x|\theta)dx};\\
\frac{\partial^2H_n}{\partial \xi\partial\theta} & = & \frac{\lambda}{1-\lambda}\int{\psi'\left(\xi^tg(x)\right)g(x)\nabla_{\theta}p_1(x|\theta)^tdx}; \\
\frac{\partial^2H_n}{\partial \alpha\partial\lambda} & = & 0; \\
\frac{\partial^2H_n}{\partial \alpha\partial\theta} & = & 0; \\
\frac{\partial^2H_n}{\partial \lambda\partial\theta} & = &  \frac{1}{(1-\lambda)^2}\int{\psi\left(\xi^tg(x)\right)\nabla_{\theta}p_1(x|\theta)dx}.
\end{eqnarray*}
As $n$ goes to infinity, we have $\bar{\xi}\rightarrow 0$ and $\bar{\phi}\rightarrow \phi^*$. Then, under regularity assumptions of the present theorem, we can calculate the limit in probability of the matrix $J_{H_n}(\bar{\phi},\bar{\xi})$. The blocks limits are given by
\[
\frac{\partial^2H_n}{\partial \xi^2} \stackrel{\mathbb{P}}{\rightarrow} - \mathbb{E}_{P_0^*}\left[g(X)g(X)^t\right],\qquad \frac{\partial^2H_n}{\partial \alpha^2} \stackrel{\mathbb{P}}{\rightarrow} 0, \qquad \frac{\partial^2H_n}{\partial \lambda^2} \stackrel{\mathbb{P}}{\rightarrow} 0, \qquad \frac{\partial^2H_n}{\partial \theta^2} \stackrel{\mathbb{P}}{\rightarrow} 0,\qquad \frac{\partial^2H_n}{\partial \xi\partial\alpha} \stackrel{\mathbb{P}}{\rightarrow} \nabla m(\alpha^*) \]

\[\frac{\partial^2H_n}{\partial \xi\partial\lambda} \stackrel{\mathbb{P}}{\rightarrow}  - \frac{1}{(1-\lambda^*)^2}\mathbb{E}_{P_T}\left[g(X)\right] + \frac{1}{(1-\lambda^*)^2}\int{g(x)p_1(x|\theta^*)dx}\]

\[\frac{\partial^2H_n}{\partial \xi\partial\theta} \stackrel{\mathbb{P}}{\rightarrow} \frac{\lambda^*}{1-\lambda^*}\int{g(x)\nabla_{\theta}p_1(x|\theta^*)dx},\qquad \frac{\partial^2H_n}{\partial \alpha\partial\lambda} \stackrel{\mathbb{P}}{\rightarrow} 0, \qquad
\frac{\partial^2H_n}{\partial \alpha\partial\theta}  \stackrel{\mathbb{P}}{\rightarrow} 0, \qquad \frac{\partial^2H_n}{\partial \lambda\partial\theta} \stackrel{\mathbb{P}}{\rightarrow}  0, \]
taking into account that $\psi(0)=0,\psi'(0)=1$ and $\psi''(0)=1$. The limit in probability of the matrix $J_{H_n}(\bar{\phi},\bar{\xi})$ can be written in the form:
\[ J_H = \left[\begin{array}{cc}
0 & J_{\phi^*,\xi^*}^t \\
J_{\phi^*,\xi^*} & J_{\xi^*,\xi^*}
\end{array}\right],\]
where $J_{\phi^*,\xi^*}$ and $J_{\xi^*,\xi^*}$ are given by (\ref{eqn:NormalAsymMomJ1}) and (\ref{eqn:NormalAsymMomJ2}). The inverse of matrix $J_H$ has the form:
\[J_H^{-1} = \left(\begin{array}{cc} -\Sigma & H \\ H^t & W\end{array}\right),\]
where
\[
\Sigma = \left(J_{\phi^*,\xi^*}^t J_{\xi^*,\xi^*} J_{\phi^*,\xi^*}\right)^{-1},\quad  H = \Sigma J_{\phi^*,\xi^*}^t J_{\xi^*,\xi^*}^{-1},\quad  W = J_{\xi^*,\xi^*}^{-1} - J_{\xi^*,\xi^*}^{-1} J_{\phi^*,\xi^*} \Sigma J_{\phi^*,\xi^*}^t J_{\xi^*,\xi^*}^{-1}.
\]
Going back to (\ref{eqn:StochExpansion}), we have:
\begin{equation*}
\left(\begin{array}{c} 0 \\ 0 \end{array}\right) = \left(\begin{array}{c}  0 \\ \frac{\partial H_n}{\partial \xi}(\phi^*,0) \end{array}\right)  + J_{H_n}(\bar{\phi},\bar{\xi}) \left(\begin{array}{c}  \hat{\phi}-\phi^* \\ \xi_n(\hat{\phi}) \end{array}\right).
\end{equation*}
Solving this equation in $\phi$ and $\xi$ gives:
\begin{equation*}
\left(\begin{array}{c}  \sqrt{n}\left(\hat{\phi}-\phi^*\right) \\ \sqrt{n}\xi_n(\hat{\phi})\end{array}\right) = J_H^{-1}\left(\begin{array}{c}  0 \\ \sqrt{n}\frac{\partial H_n}{\partial \xi}(\phi^*,0) \end{array}\right) + o_P(1).
\end{equation*}
Finally, using (\ref{eqn:LimitLawPartialDerivHn}), we get that:
\[\left(\begin{array}{c}  \sqrt{n}\left(\hat{\phi}-\phi^*\right) \\ \sqrt{n}\xi_n(\hat{\phi})\end{array}\right) \xrightarrow[\mathcal{L}]{} \mathcal{N}\left(0,S\right)\]
where 
\[S = \frac{1}{(1-\lambda^*)^2}\left(\begin{array}{c}H \\ W\end{array}\right) \text{Var}_{P_T}(g(X)) \left(H^t\quad W^t\right).\]
This ends the proof.
\end{proof}

\bibliographystyle{IEEEtran}
\bibliography{IEEEabrv,bibliography-paper}

\end{document}